\documentclass[journal,12pt,draftclsnofoot,onecolumn]{IEEEtran}
\usepackage{amsmath,amsthm,graphicx,amssymb,psfrag}
\newtheorem{Theorem}{Theorem}
\newtheorem{proposition}[Theorem]{Proposition}
\newtheorem{corollary}[Theorem]{Corollary}

\def\E{{\mathbb{E}}}
\def\diag{{\mbox{diag}}}

\newcommand{\ud}{\,\mathrm{d}}

\newcommand{\evec}{{\bf{e}}}

\newcommand{\yvec}{{\bf{y}}}

\newcommand{\xvec}{{\bf{x}}}
\newcommand{\zvec}{{\bf{z}}}

\newcommand{\svec}{{\bf{s}}}

\newcommand{\gvec}{{\bf{g}}}

\newcommand{\zerovec}{{\bf{0}}}

\newcommand{\Gammamat}{{\bf{\Gamma}}}
\newcommand{\Omegamat}{{\bf{\Omega}}}

\newcommand{\Amat}{{\bf{A}}}
\newcommand{\Bmat}{{\bf{B}}}
\newcommand{\Cmat}{{\bf{C}}}
\newcommand{\Dmat}{{\bf{D}}}

\newcommand{\Jmat}{{\bf{J}}}

\newcommand{\Imat}{{\bf{I}}}

\newcommand{\Wmat}{{\bf{W}}}
\newcommand{\Xmat}{{\bf{X}}}

\newcommand{\Zmat}{{\bf{Z}}}

\newcommand{\define}{\stackrel{\triangle}{=}}

\def\thetavec{{\mbox{\boldmath $\theta$}}}
\def\phivec{{\mbox{\boldmath $\phi$}}}

\def\thetavecsmall{{\mbox{\boldmath {\scriptsize $\theta$}}}}

\def\phivecsmall{{\mbox{\boldmath {\scriptsize $\phi$}}}}

\def\phivecsmall{{\mbox{\boldmath {\scriptsize $\phi$}}}}

\newcommand{\be}{\begin{equation}}
\newcommand{\ee}{\end{equation}}
\newcommand{\beqna}{\begin{eqnarray}}
\newcommand{\eeqna}{\end{eqnarray}}

\title{Graph Signal Processing Meets Blind Source Separation}

\author{{Jari~Miettinen, {\it Member, IEEE}, Eyal Nitzan, {\it Member, IEEE},  Sergiy~A.~Vorobyov, {\it Fellow, IEEE}, and~Esa Ollila, {\it Senior Member, IEEE}}
\thanks{ This paper was presented in part at the IEEE International Conference on Acoustics, Speech, and Signal Processing, Virtual Barcelona, Spain, May 4--8, 2020 \cite{ICASSP2020}. 
	
The authors are with Department of Signal Processing and Acoustics, Aalto University,
P.O. Box 15400, FIN-00076 Aalto, Finland (emails: jari.p.miettinen@aalto.fi, eyal.nitzan@aalto.fi, sergiy.vorobyov@aalto.fi, esa.ollila@aalto.fi).

This research was supported in parts by the Academy of Finland grants No.~299243 and No.~298118.}}

\begin{document}
\maketitle
\begin{abstract}
%150-250 words
In graph signal processing (GSP), prior information on the dependencies in the signal is collected in a graph which is then used when processing or analyzing the signal. Blind source separation (BSS) techniques have been developed and analyzed in different domains, but for graph signals the research on BSS is still in its infancy. In this paper, this gap is filled with two contributions. First, a nonparametric BSS method, which is relevant to the GSP framework, is refined, the Cram\'{e}r-Rao bound (CRB) for mixing and unmixing matrix estimators in the case of Gaussian moving average graph signals is derived, and for studying the achievability of the CRB, a new parametric method for BSS of Gaussian moving average graph signals is introduced. Second, we also consider BSS of non-Gaussian graph signals and two methods are proposed. Identifiability conditions show that utilizing both graph structure and non-Gaussianity provides a more robust approach than methods which are based on only either graph dependencies or non-Gaussianity. It is also demonstrated by numerical study that the proposed methods are more efficient in separating non-Gaussian graph signals.
\end{abstract}


\IEEEpeerreviewmaketitle

\section{INTRODUCTION}
\label{sec:intro}
In blind source separation (BSS) models, the data matrix is a (usually linear and instantaneous) mixture of latent components, which are assumed to be mutually independent~\cite{ICABook2001,ComonJutten2010}. The purpose of transforming the data matrix into a matrix with independent components is to help in analyzing, learning or visualizing the data. If the data points are independently and identically distributed (i.i.d.), the only way to solve the BSS problem is to use the probability distributions of the components. This branch is known as independent component analysis (ICA), even though sometimes the term ICA is also used as a substitute for BSS. Some well-known ICA estimators are joint approximate diagonalization (JADE)~\cite{CardosoSoloumiac1993}, Infomax~\cite{BellSejnowski1995} and FastICA~\cite{HyvarinenOja1997,Hyvarinen1999}. All these methods require that at most one of the independent components has a Gaussian density. The reason why non-Gaussianity can be used in linear BSS is that a weighted sum of non-Gaussian random variables is more Gaussian than the summands due to the central limit theorem. 

Often the data has some kind of structure, which the ICA methods do not utilize. The structure might arise from serial or spatial dependence, and from tensor or graph data structures. In case of tensor-valued data, the BSS model is different than in other domains and it is important to take the structure into account instead of reducing it to multiple basic BSS models to effectively solve the problem~\cite{Virtaetal2017}. Also in case of time series, the use of temporal dependence has the advantage that Gaussian components can be separated as well. Such methods include second-order blind identification (SOBI)~\cite{Belouchranietal1997}, which uses autocovariance matrices, and color ICA~\cite{Leeetal2011}, which uses Whittle likelihood. Using both non-Gaussianity and temporal dependence have been proposed in~\cite{Shietal2009}.

Recently, the increasing volume of signals with irregular dependencies, coming for example from social, sensor or gene regulatory networks, has lead to the growth of the field of graph signal processing (GSP). There are many signal processing techniques that have been already extended to graph signals, see, e.g., \cite{SandryhailaMoura2013, Shumanetal2013}. In GSP, prior information on the dependence between the data points, which is given as a graph, is used in analyzing the data. One existing method that can be viewed as a BSS method for graph signals (but which was not designed within the GSP framework), called graph decorrelation (GraDe), was proposed in~\cite{Blochletal2010} as a tool for gene expression data analysis~\cite{Blochletal2010, Kowarschetal2010}. It can clearly be useful with some modifications in spatial ICA~\cite{McKeownetal1998} for fMRI data as well. More uses can be found in other biological applications, social networks, sensor networks, etc. Additionally, there is a work under the term graph BSS in the estimation of states and topology of power systems~\cite{Tirza}. However, there the relation to graphs is that the mixing matrix is assumed to be an unknown graph Laplacian matrix, whereas in our BSS model the known graphs give the dependence structure of the source components. To the best of our knowledge, the existing literature of BSS of graph signals is limited to~\cite{Blochletal2010, Kowarschetal2010}. Therefore, a more comprehensive study of BSS of graph signals starting from the basic principles is of significant interest. 

\subsection{Contributions}
The contributions of this paper are the following.\footnote{Our preliminary work on the topic has been reported in \cite{ICASSP2020}.} \\
(i) BSS problem for non-Gaussian graph signals is formulated and identifiability conditions for mixing matrix estimation are presented. \\
(ii) Graph decorrelation procedure is revised by using graph autocorrelation matrices, i.e., normalized graph autocovariance matrices, which are computed for multiple adjacency matrices taking into account the possibility that the graph structures may differ for different independent components (ICs). \\
(iii) Maximum likelihood (ML) approach for separation of graph moving average signals is studied. \\
(iv) Cram\'{e}r-Rao bound (CRB) for vectorized mixing and unmixing matrix estimators are derived in the case of Gaussian graph moving average (GMA) signal sources. The existing results on bounds for mixing and unmixing matrix estimation in the BSS context include the closed-form CRB for the mixing and unmixing matrices derived for the i.i.d case in~\cite{ESA_CRB}, and the CRB for the case when the mixing matrix is specifically a Laplacian matrix~\cite{Tirza}. In addition, Cram\'{e}r-Rao-type bounds on the interference-to-signal ratio for autoregressive moving average time series processes and for BSS of Gaussian sources with general/arbitrary (unspecific) covariance structures were derived in \cite{YEREDOR_CRB} and \cite{YEREDOR_GAUSSIAN_CRB}, respectively. In our case, the covariance is structured by the graph adjacency matrix, which is the major difference to the existing results. \\
(v) Two BSS methods for non-Gaussian graph signals are proposed. The first method jointly diagonalizes graph autocorrelation matrices and fourth-order cumulant matrices. Concerning the second method, a new fixed-point algorithm for joint graph decorrelation and FastICA is derived.

\subsection{Notation} 
We use boldface capital letters for matrices, boldface lowercase letters for vectors, and capital calligraphic letters for sets. The exceptions are $\mathbf{1}_N$ that is the $N$-dimensional vector full of ones and the $M \times N$ matrix full of ones $\mathbf{1}_{M\times N} = \mathbf{1}_M \mathbf{1}_N^\top$. The matrix $\mathbf{I}_N$ is the $N \times N$ identity matrix, $\mathbf{e}_i$ is the $i$th unit vector, and $\mathbf{E}_{P \times P}^{i,j}$ is the $P \times P$ matrix whose $(i,j)$th element is one and the other elements are zeros. The $(i,j)$th element of the matrix $\Amat$ is denoted by $a_{i,j}$. The Kronecker delta is denoted as $\delta(j-l)$.  The notations $(\cdot)^\top$,  $\odot$, $\| \cdot \|$,  $\mbox{tr} ( \cdot )$, $\mbox{det} (\cdot)$, $\mathbb{P} (\cdot)$, $\E \{ \cdot \}$, $\text{diag}(\cdot)$, and $\text{vec}(\cdot)$ stand for the transpose, Hadamard product, Euclidean norm of a vector, trace of a matrix, determinant of a matrix, probability, mathematical expectation, diagonal elements of a matrix (or a diagonal matrix formed from a vector), and a vector of elements of a matrix obtained by stacking the columns one after another, respectively. The matrix inequality $\Amat\succeq\Bmat$ means that $\Amat-\Bmat$ is a positive semidefinite matrix. In addition, ${\mathcal{N}}(\zerovec, \mathbf{C})$ stands for Gaussian distribution of a vector random variable with zero mean and covariance matrix $\mathbf{C}$.

\subsection{Organization}
In Section~\ref{sec:BSS}, we recall the BSS model and the graph decorrelation method to estimate the unmixing matrix in case of graph signal sources, and modify the graph decorrelation method. For Gaussian GMA model, we derive the CRB for the mixing and unmixing matrix estimators in Section~\ref{sec:CRB}. Non-Gaussian graph signals are considered in Section~\ref{sec:NGGBSS}, where we introduce two methods for BSS of graph signals which take advantage of non-Gaussianity as well. Section~\ref{sec:simul} collects the simulation results related to the previous sections, but also includes a study of ML approach in estimating the unmixing matrix under the assumption that the sources come from the Gaussian GMA model. Section~\ref{sec:concl} concludes the paper.

\section{Gaussian Graph Signal BSS}
\label{sec:BSS}
%We first recall graph and graph signal models which will be used later, and formulate the graph BSS problem. 
A graph $\mathcal{G} = (\mathcal{N},\mathcal{E})$ consists of a set of $N$ nodes $\mathcal{N}$ and a set of edges $\mathcal{E}$. It can be presented concisely by an $N\times N$ adjacency matrix $\mathbf{W}$ which satisfies the conditions $w_{ii}= 0$ for $i=1,\dots,N$ and $w_{ij} = 1$ if and only if $(i,j) \in \mathcal{E}$, i.e., there is an edge from the $j$th node to the $i$th node. We will assume for simplicity that the considered signal dependency graph is unweighted and undirected. A graph is undirected if $(i,j) \in \mathcal{E}$ implies $(j,i) \in \mathcal{E}$. The adjacency matrix of an undirected graph is symmetric. Most of the results in the paper can be extended to weighted and directed case as well.

\subsection{Graph BSS Problem}
Let us assume that the observed (centered) signal matrix $\mathbf{X} \in \mathbb{R}^{P\times N}$ is a linear mixture of the signals/components of $\mathbf{Z}$, that is, 
\begin{equation}
\label{BSS_model}
\mathbf{X} = \mathbf{\Omega} \mathbf{Z}
\end{equation}
where the mixing matrix $\mathbf{\Omega} \in \mathbb{R}^{P\times P}$ has full rank, and $\mathbf{Z} \define (\boldsymbol z_1, \dots, \boldsymbol z_P)^\top \in \mathbb{R}^{P\times N}$ is a matrix of latent source components. 
It is assumed that the components of $\mathbf{Z}$ are mutually independent, i.e., are ICs. In addition, it is assumed that $\E[\mathbf{z}_p] = 0$, $\E[ \| \mathbf{z}_p \|^2] = N$, and hence, 
$N^{-1} {\rm tr} ( \E[ \mathbf{z}_p \mathbf{z}_p^\top]) = 1$, for $p=1,\dots,P$. These assumptions are needed, because it is a well-known fact that the scales of the source components and columns of the mixing matrix cannot be estimated jointly.\footnote{This can be seen also in the CRB derivations in Section~\ref{sec:CRB}, since if the scales of the sources were unknown, the Fisher information matrix would be singular.}

For example in brain imaging applications, the matrix $\mathbf{Z}$ includes $P$ brain signals of interest, which cannot be directly measured with non-invasive methods. We can only have observations on the mixtures of these signals, given in $\mathbf{X}$. The unknown mixing matrix $\mathbf{\Omega}$ describes how the latent signals are mixed on the way to sensors. See Introduction for more applications of the model.

To ensure identifiability of the unmixing matrix $\mathbf{\Gamma} \define \mathbf{\Omega}^{-1}$ based on the knowledge of $\mathbf{X}$ only, the pairs of components of $\mathbf{Z}$ have to differ by graph dependence properties or at least one of the two must be non-Gaussian.\footnote{More detailed  identifiability conditions will be given when it is possible, that is, for the particular estimators which will be proposed in Section~\ref{sec:NGGBSS}.} We first consider the case when all ICs are Gaussian and differ from each other only by graph dependence properties.

Typically, the first step in BSS is prewhitening, that is,
\begin{align}
\tilde{\mathbf{X}} = \hat{\mathbf{S}}_0^{-1/2}\mathbf{X},
\end{align}
where $\hat{\mathbf{S}}_0=N^{-1}\mathbf{X}\mathbf{X}^\top$ is the sample covariance matrix of $\mathbf{X}$. Since $\mathbf{Z} = \mathbf{U} (\mathbf{\Omega} \mathbf{\Omega}^\top)^{-1/2 }\mathbf{X}$ for an orthogonal matrix $\mathbf{U} \define (\mathbf{u}_1, \dots, \mathbf{u}_P)^\top$ and $\hat{\mathbf{S}}_0$ is an estimate of $\mathbf{\Omega}\mathbf{\Omega}^\top$, we have $\mathbf{Z} \approx \mathbf{U} \tilde{\mathbf{X}}$. The aim after the prewhitening is then to find an estimate $\hat{\mathbf{U}}$ such that $\hat{\mathbf{U}}\tilde{\mathbf{X}}$ has ICs, and the unmixing matrix estimate is $\hat{\mathbf{\Gamma}}=\hat{\mathbf{U}}\hat{\mathbf{S}}_0^{-1/2}$. 

\subsection{GraDe}
When prewhitening is used, BSS estimator can be defined as an optimizer of an objective function $f(\mathbf{U},\tilde{\mathbf{X}})$ with respect to $\mathbf{U}$ under the orthogonality constraint $\mathbf{U} \mathbf{U}^\top = \mathbf{I}_P$. There are also adaptive BSS methods where the choice of objective itself depends on $\tilde{\mathbf{X}}$. If ICs are Gaussian and differ from each other only by graph dependence properties, it is reasonable to choose as an objective function, denoted hereafter as  $f_1 (\mathbf{U}, \tilde{\mathbf{X}})$, the graph decorrelation objective, which uses the prior information on the connections between the elements in each row of $\mathbf{Z}$. In the simplest case, this information is given in the form of a single symmetric $N\times N$ adjacency matrix $\mathbf{W}$, which is used in defining graph autocovariance matrices  
\begin{align}
\label{GraDemat}
\hat{\mathbf{S}}^k (\tilde{\mathbf{X}}, \mathbf{W}) = \frac{1}{N-k} (\tilde{\mathbf{X}} \mathbf{W}^k \tilde{\mathbf{X}}^\top), \ \ k=1,\dots,K.
\end{align}

Then the estimator looks for an orthogonal matrix $\mathbf{U}$ which makes the matrices $\mathbf{U} \hat{\mathbf{S}}^k (\tilde{\mathbf{X}}, \mathbf{W}) \mathbf{U}^\top$ as diagonal as possible. Thus, the objective function to be maximized with respect to $\mathbf{U}$ is $\sum_{k=1}^K\| \diag(\mathbf{U} \hat{\mathbf{S}}^k (\tilde{\mathbf{X}}, \mathbf{W}) \mathbf{U}^\top) \|^2$. Notice that this is equivalent to minimizing the sum of squares of the off-diagonal elements. When $K=1$, the corresponding optimization problem can be solved simply in terms of eigendecomposition, whereas for $K\geq 2$, it becomes the joint approximate diagonalization problem which is computationally more demanding. However, joint diagonalization algorithm~\cite{Clarkson1988} based on Givens rotations, for example, yields a very fast solution. It is used in JADE and many other BSS methods.  

\subsection{Extensions of GraDe}

We propose two modifications/extensions to the graph decorrelation principle described above. First, we allow several adjacency matrices $\mathbf{W}_1,\dots,\mathbf{W}_P$ and their powers ${\cal W} \define \{\mathbf{W}_p^k\}_{p=1,\dots,P,\ k=1,\dots,K}$. This is useful if the ICs have different dependence structures, which can be captured by distinct adjacency matrices. An example of such an application is the spatial ICA for fMRI data, where brain regions are interacting differently in different brain functions. Note that if there are irrelevant adjacency matrices in the set $\{\mathbf{W}_1,\dots,\mathbf{W}_P\}$, i.e., there are matrices for which the graph autocorrelations (see below) are equal for each IC, their influence on $\mathbf{U}$ is quite small. We will illustrate this via simulations in Section~\ref{sec:simul}.

The second modification aims to keep a balance between the weights that different matrices $\mathbf{W}_p^k$ may have in the overall objective function, which can now be written as 
\begin{align}
\label{GraDeobj}
f_1 (\mathbf{U}, \tilde{\mathbf{X}}) \define \sum_{p=1}^P \sum_{k=1}^K\| \diag(\mathbf{U} \tilde{\mathbf{S}}_p^k (\tilde{\mathbf{X}}, \mathbf{W}) \mathbf{U}^\top) \|^2
\end{align} 
where the balance can be ensured in terms of normalization, i.e., by using the normalized matrices 
\begin{equation}
\tilde{\mathbf{S}}_p^k (\tilde{\mathbf{X}}, \mathbf{W}) = \frac{P \tilde{\mathbf{X}} \mathbf{W}_p^k \tilde{\mathbf{X}}^\top}{\|\mathbf{W}_p^k \tilde{\mathbf{X}}^\top\|}
\end{equation}  
instead of the graph autocovariance matrices \eqref{GraDemat}. These normalized matrices can be interpreted as graph autocorrelation matrices (cf. autocorrelation matrix of multivariate time series process), and they are invariant with respect to rescaling $\mathbf{W}_p$. The latter is an important property when using multiple adjacency matrices or when using them as part of a composite objective function \eqref{genObjF}, which will be discussed in Section~\ref{sec:NGGBSS}, because then the same weight can be applied for the graph autocorrelation part irrespective of the number of edges in the graphs. 

It should be noted that no particular graph signal model needs to be assumed for applying the graph decorrelation principle. However, parametric graph signal models are often used in GSP as they lead to desired higher accuracy of GSP methods (if the model is adequate to the graph signals) just as in classical signal processing~\cite{Blochletal2010, Chepuri2017}. Given a parametric model, it is of wide interest to derive the CRB for performance analysis, as will be presented in the following section.

\section{CRB for Gaussian Graph Signal BSS}
\label{sec:CRB}
In this section, we assume a parametric graph signal model and derive the CRB for unbiased estimation of the mixing and unmixing matrices in the mixing model \eqref{BSS_model}. In particular, the sources are assumed to be Gaussian with zero mean and their covariance matrices depend on a corresponding adjacency matrix and on some additional unknown parameters. Formally, we consider $P$ statistically independent sources, $\zvec_p \sim {\mathcal{N}} (\zerovec, \Cmat_p)$, where $\Cmat_p=\Cmat_p(\Wmat_p,\thetavec_p)$ are invertible matrices that depend on the known adjacency matrix, $\Wmat_p$, and on some unknown parameters, $\thetavec_p\in\mathbb{R}^{M_p},~p=1, \ldots, P$.\\
\indent
It can be observed that there are no specific assumptions on the structure of the covariance matrices, $\Cmat_p,~p=1,\ldots,P$. In order to demonstrate the relevance of the CRB to separation of graph signals, we present a common GSP parametric model, which we term hereafter as GMA model. The GMA model has been widely used~\cite{Blochletal2010, Chepuri2017}, and therefore, it is of interest to show the corresponding CRB for this case. A GMA model of order $M$, GMA$(M)$, is written as 
\begin{equation}
\mathbf{z} = \mathbf{y} + \sum_{l=1}^M \theta_l \mathbf{W}^l \mathbf{y} \label{GMA}
\end{equation}
where the elements of the innovation vector $\mathbf{y} \define [ y_1, \dots, y_N ]^\top$ are independent identically distributed (i.i.d.) random variables with variance $\sigma^2$, while the parameters, $\theta_1, \dots, \theta_M$, adjust the dependencies between the nodes in the graph signal. The model \eqref{GMA} is a natural extension of the time series MA model, and it can be seen that in addition to the parameters, $\theta_1, \dots, \theta_M$, there is also the adjacency matrix parameter $\mathbf{W}$. The adjacency matrix is usually assumed to be known, while the parameters, $\theta_1, \dots, \theta_M$, may be also known or unknown. It is worth noting that GMA signals are stationary graph processes~\cite{Girault2015, PerraudinVandergheynst2017, Marquesetal2016}.\\
\indent
We begin the analysis by describing the general parametric BSS model for Gaussian graph signals. Then, in Subsection \ref{subsec:Derivation of the CRB} we derive the CRB for the general case and afterwards we present the CRB for Gaussian GMA signals, which demonstrates the applicability of the bound to graph signals BSS. Let us define $M\define\sum_{p=1}^{P}M_p$, the vector $\thetavec\define[\thetavec_1^\top, \ldots, \thetavec_P^\top]^\top$, the matrices
\be\label{C_m_deriv}
\Dmat_{p,m} \define \frac{\partial\Cmat_p} {\partial\theta_{p,m}},~m=1,\ldots,M_p,~p=1, \ldots, P,
\ee
and the vectors
\be\label{s_m_deriv}
\svec_{p} \define [{\rm{tr}}(\Cmat_p^{-1} \Dmat_{p,1}),\ldots,{\rm{tr}}(\Cmat_p^{-1} \Dmat_{p,M_p})]^\top,~p=1, \ldots, P,
\ee
that will be used for future derivations. In addition, we define
\be\label{kappa_define}
\kappa_{i,j}\define{\rm{tr}}\left(\Cmat_j^{-1}\Cmat_i\right) , \quad i,j = 1, \ldots, P.
\ee
\indent
Let $\xvec\define{\text{vec}}(\Xmat)$ and $\zvec\define{\text{vec}}(\Zmat)$, where $\Xmat$ and $\Zmat$ are defined in \eqref{BSS_model}. Then, the mixing model \eqref{BSS_model} can be written in a vector form as
\be\label{vec_BSS}
\xvec = (\Imat_N \otimes \Omegamat) \zvec.
\ee
Thus, it can be seen that the vector $\xvec$ is a Gaussian random vector, whose distribution is ${\mathcal{N}}(\zerovec, \Cmat_\xvec)$, where
\begin{equation*}
\Cmat_\xvec\define(\Imat_N\otimes\Omegamat)\Cmat_\zvec(\Imat_N\otimes\Omegamat^\top)
\end{equation*}
and
\begin{equation*}
\Cmat_\zvec \define \sum_{p=1}^{P} \Cmat_p \otimes (\evec_p \evec_p^\top)
\end{equation*}
is obtained using the statistical independence of the zero-mean sources. 
The unknown parameter vector in the model from \eqref{vec_BSS} is
\be\label{parameter_vector}
\phivec \define [{\text{vec}}^\top (\Omegamat), \thetavec^\top]^\top \in \mathbb{R}^{P^2 + M}.
\ee
In the BSS context, we are interested in the estimation of ${\text{vec}}(\Omegamat)$, where $\thetavec$ are unknown nuisance parameters. 
%However, in some cases estimation of $\thetavec$ may be of interest as well. 
It should be noted that some of the derivations in this section are similar, but not identical, to the derivations in \cite{YEREDOR_GAUSSIAN_CRB} for BSS in the case of Gaussian sources with general covariance structures. For completeness the full derivations are given in the paper.
%%%%%%%%%%%%%%%%%%%%%%%%%%%%%%%%%%%%%%%%%%%%%%%%%%%%%%%%%%%%%%%%%%%%%%%%%%%%%%%%
%%%%%%%%%%%%%%%%%%%%%%%%%%%%%%%%%%%%%%%%%%%%%%%%%%%%%%%%%%%%%%%%%%%%%%%%%%%%%%%%
%%%%%%%%%%%%%%%%%%%%%%%%%%%%%%%%%%%%%%%%%%%%%%%%%%%%%%%%%%%%%%%%%%%%%%%%%%%%%%%%
\subsection{Derivation of the CRB}\label{subsec:Derivation of the CRB}
Under the model \eqref{vec_BSS} with the parameter vector from \eqref{parameter_vector}, the Fisher information matrix (FIM), $\Jmat$, is a $(P^2 + M)\times(P^2 + M)$ matrix that can be expressed in a block form as 
\begin{equation*}
\Jmat=\begin{bmatrix}
\Jmat_\Omegamat & \Jmat_{\Omegamat,\thetavecsmall} \\
\Jmat_{\Omegamat,\thetavecsmall}^\top & \Jmat_\thetavecsmall
\end{bmatrix}
\end{equation*}
where $\Jmat_\Omegamat$, $\Jmat_\thetavecsmall$, and $\Jmat_{\Omegamat,\thetavecsmall}$ are the blocks corresponding to ${\text{vec}}(\Omegamat)$, $\thetavec$, and the coupling between ${\text{vec}}(\Omegamat)$ and $\thetavec$, respectively.\\
\indent
In the following theorem, we derive the FIM for estimation of $\phivec$ from \eqref{parameter_vector}. 
\begin{Theorem}\label{Th_FIM}
Under the model \eqref{vec_BSS} with the parameter vector \eqref{parameter_vector}, the blocks of the FIM are given as follows:\\
The matrix $\Jmat_\Omegamat$ is a block matrix, whose $(i,j)$th block is a $P\times P$ matrix given by
\be\label{FIM_compute_Omega_block}
\begin{split}
&\Jmat_\Omegamat^{(i,j)}=\begin{cases}
\Omegamat^{-\top} \left( 2N\evec_i\evec_i^\top+\sum_{\underset{l\neq i}{l=1}}^{P}\kappa_{i,l}\evec_l\evec_l^\top \right) \Omegamat^{-1}, \quad i=j \\
\Omegamat^{-\top}N\evec_j\evec_i^\top\Omegamat^{-1}, \quad i\neq j
\end{cases}, \quad i,j=1,\ldots,P.
\end{split}
\ee
The matrix $\Jmat_\thetavecsmall$ is a block diagonal matrix, whose $p$th diagonal block is an $M_p\times M_p$ matrix denoted by $\Jmat_{\thetavecsmall_p}$, where
\be\label{FIM_compute_theta_sub}
j_{\thetavecsmall_p,i,j}=\frac{1}{2}{\rm{tr}}\left(\Cmat_p^{-1}\Dmat_{p,i}\Cmat_p^{-1}\Dmat_{p,j}\right),~i,j=1,\ldots,M_p.
\ee 
The matrix $\Jmat_{\Omegamat,\thetavecsmall}$ is a block diagonal matrix, whose $p$th diagonal block is a $P\times M_p$ matrix given by
\be\label{FIM_compute_Omega_theta_block}
\begin{split}
\Jmat_{\Omegamat,\thetavecsmall}^{(p)}=\Omegamat^{-\top}\evec_p\svec_p^\top, \quad p = 1, \ldots, P.
\end{split}
\ee
\end{Theorem}
\begin{proof}
The proof is given in Appendix  \ref{App_Th_FIM}.	
\end{proof} 
\noindent
Similar results to \eqref{FIM_compute_Omega_block} and \eqref{FIM_compute_Omega_theta_block} have been derived in \cite{YEREDOR_GAUSSIAN_CRB} for estimation of the unmixing matrix in the case of Gaussian sources with general covariance structures.\\
\indent
The CRB for estimation of $\Omegamat$ or equivalently ${\text{vec}}(\Omegamat)$ is denoted by $\textbf{CRB}_{\Omegamat}$. This bound is a $P^2\times P^2$ matrix that can be expressed as a block matrix. The $(i,j)$th block of $\textbf{CRB}_{\Omegamat}$ is a $P\times P$ matrix and is denoted by $(\textbf{CRB}_{\Omegamat})^{(i,j)},~i,j=1,\ldots,P$. In the following proposition, we derive $\textbf{CRB}_{\Omegamat}$.
\begin{proposition}\label{prop_CRB_omega}
Given that the FIM from Theorem \ref{Th_FIM} is an invertible matrix, the $(i,j)$th block of $\textbf{CRB}_{\Omegamat}$ is given by
\be\label{CRB_compute_Omega_block}
\begin{split}
&(\textbf{CRB}_{\Omegamat})^{(i,j)}
=\begin{cases}
\Omegamat \left( \frac{1}{\zeta_i}\evec_i\evec_i^\top + \sum_{\underset{l\neq i}{l=1}}^{P} \frac{\kappa_{l,i}}{\kappa_{i,l} \kappa_{l,i}-N^2}\evec_l\evec_l^\top \right) \Omegamat^{\top}, \quad i=j\\
\Omegamat \left( -\frac{N}{\kappa_{i,j}\kappa_{j,i}-N^2}\evec_j\evec_i^\top \right) \Omegamat^{\top}, \quad i\neq j
\end{cases}, \quad i,j=1,\ldots,P,
\end{split}
\ee 
where
\be\label{zeta_general}
\zeta_p\define 2N-\svec_p^\top\Jmat_{\thetavecsmall_p}^{-1}\svec_p.
\ee
\end{proposition}
\begin{proof}
The proof is given in Appendix \ref{App_prop_CRB_omega}.
\end{proof}
\noindent
It can be seen from \eqref{kappa_define} that for a pair of ICs $\zvec_i$ and $\zvec_j$ with the same covariance matrix up to a scale, i.e., $\Cmat_i = \rho\Cmat_j,~\rho>0$, we have $\kappa_{i,j}\kappa_{j,i}=N^2$. Thus, the CRB does not exist in accordance with \eqref{CRB_compute_Omega_block}. The reason for this result is that in this case, Gaussian ICs  $\zvec_i$ and $\zvec_j$ cannot be separated \cite{YEREDOR_GAUSSIAN_CRB}.\\
\indent
Let $\textbf{CRB}_{\Imat_P}$ denote $\textbf{CRB}_{\Omegamat}|_{\Omegamat=\Imat_P}$. Then, using \eqref{CRB_compute_Omega_block}, we can write
\be\label{equivariance_Omega}
\textbf{CRB}_{\Omegamat}=(\Imat_P\otimes\Omegamat)\textbf{CRB}_{\Imat_P}(\Imat_P\otimes\Omegamat^\top).
\ee
The result \eqref{equivariance_Omega} demonstrates that the equivariance property holds for our CRB as expected. It is, for example, also the case for the CRBs derived in \cite{ESA_CRB,YEREDOR_CRB,YEREDOR_GAUSSIAN_CRB}. Moreover, it is worth noting that the same equivariance property holds for common BSS algorithms (see e.g. \cite{miettinen2016separation}) as well. In addition, it can be observed that $\textbf{CRB}_{\Imat_P}$ is composed of $\frac{P(P-1)}{2}$ CRBs derived for all combinations of pairs of sources.\\
\indent
If $\thetavec$ is known, the resulting CRB for estimation of $\Omegamat$, denoted by ${\textbf{CRB}}_{\Omegamat}^{(*)}\in\mathbb{R}^{P^2\times P^2}$, is obtained in a similar manner to \eqref{CRB_compute_Omega_block}. The $(i,j)$th block of ${\textbf{CRB}}_{\Omegamat}^{(*)}$ is given by
\be\label{CRB_compute_Omega_block_theta_known}
\begin{split}
&({\textbf{CRB}}_{\Omegamat}^{(*)})^{(i,j)}=(\Jmat_\Omegamat^{-1})^{(i,j)} =\begin{cases}
\Omegamat \left( \frac{1}{2N}\evec_i\evec_i^\top + \sum_{\underset{l\neq i}{l=1}}^{P} \frac{\kappa_{l,i}} {\kappa_{i,l}\kappa_{l,i}-N^2}\evec_l\evec_l^\top \right) \Omegamat^{\top}, \quad i=j \\
\Omegamat \left( -\frac{N}{\kappa_{i,j}\kappa_{j,i}-N^2}\evec_j\evec_i^\top \right) \Omegamat^{\top}, \quad i\neq j
\end{cases}, \\
&\qquad i,j = 1, \ldots, P.
\end{split}
\ee 
It can be verified that $\frac{1}{\zeta_p}\geq\frac{1}{2N},~\forall p=1,\ldots,P$, and therefore, ${\textbf{CRB}}_{\Omegamat}\succeq {\textbf{CRB}}_{\Omegamat}^{(*)}$. This result is due to the coupling in the FIM between ${\text{vec}}(\Omegamat)$ and $\thetavec$. A similar result to \eqref{CRB_compute_Omega_block_theta_known} for estimation of the unmiximg matrix can be found in \cite{YEREDOR_GAUSSIAN_CRB} in the case that there is no coupling between $\thetavec$ and the unmixing matrix.\\
\indent
In the following corollary, we derive the CRB for estimation of the unmixing matrix, $\Gammamat=\Omegamat^{-1}$. We denote by $\textbf{CRB}_{\Gammamat}$ the CRB for an unbiased estimation of the vectorized unmixing matrix, $\Gammamat$. 
\begin{corollary}\label{coro_CRB_gamma}
Given $\textbf{CRB}_{\Omegamat}$ from \eqref{CRB_compute_Omega_block}, the CRB for estimation of $\Gammamat$ is given by
\be\label{CRB_Gammamat} 
\textbf{CRB}_{\Gammamat}=(\Gammamat^{\top} \otimes\Imat_P)\textbf{CRB}_{\Imat_P}(\Gammamat\otimes\Imat_P)
\ee
where $\textbf{CRB}_{\Imat_P}$ is obtained by substituting $\Omegamat=\Imat_P$ in \eqref{CRB_compute_Omega_block}.
\end{corollary}
\begin{proof}
The proof is given in Appendix  \ref{App_coro_CRB_gamma}.	
\end{proof}
\noindent
The result in Corollary~\ref{coro_CRB_gamma} demonstrates that the equivariance property also holds for the CRB for unbiased estimation of the unmixing matrix.
%%%%%%%%%%%%%%%%%%%%%%%%%%%%%%%%%%%%%%%%%%%%%%%%%%%%%%%%%%%%%%%%%%%%%%%%%%%%%%%%
%%%%%%%%%%%%%%%%%%%%%%%%%%%%%%%%%%%%%%%%%%%%%%%%%%%%%%%%%%%%%%%%%%%%%%%%%%%%%%%%
\subsection{CRB for Gaussian GMA sources}\label{subsec:CRB for GMA signals} 
In this subsection, we present the CRB from \eqref{CRB_compute_Omega_block} for estimation of the mixing matrix in the case of statistically independent Gaussian GMA sources. For brevity of the presentation, we consider $P$ GMA$(1)$ sources, $\zvec_p = (\Imat_N + \theta_p \Wmat_p) \yvec_p, \quad p = 1, \ldots, P$, where $\yvec_p \sim{\mathcal{N}} (\zerovec, \sigma_p^2 \Imat_N)$. Therefore, $\zvec_p \sim {\mathcal{N}} (\zerovec, \Cmat_p)$, where the covariance matrices are
\be\label{C_m_define}
\Cmat_p\ = \sigma_p^2(\Imat_N+\theta_p \Wmat_p) (\Imat_N+\theta_p\Wmat_p)^\top, \quad p=1, \ldots, P.
\ee
where $M_1=\ldots=M_p=1$.\\
\indent
Apart from the mixing matrix, the unknown parameters in this setting are $[\theta_1,\sigma_1^2,\ldots,\theta_P,\sigma_P^2]$. Unfortunately, such model leads to ambiguity in the scale of the sources resulting in a singular FIM \cite{ESA_CRB,YEREDOR_CRB,SLOCK}. In order to overcome this issue, we assume a parametric model for the variances such that $\sigma_p^2=\sigma_p^2 (\theta_p)$ is a known function of the unknown parameter, $\theta_p,~\forall p = 1, \ldots, P$. Thus, the unknown nuisance parameters are $\thetavec=[\theta_1,\ldots,\theta_P]^\top$.\\
\indent
Adjusting \eqref{C_m_deriv} for the considered setting, we define
\be\label{C_m_deriv_GMA}
\Dmat_p \define \frac{\ud\Cmat_p} {\ud\theta_p}. 
\ee
By substituting \eqref{C_m_define} into \eqref{C_m_deriv_GMA} with $\sigma_p^2=\sigma_p^2 (\theta_p)$, one obtains
\be\label{S_m_new_theta_depend}
\begin{split}
&\Dmat_p = \sigma_p^2 (\theta_p) (\Wmat_p + \Wmat_p^\top + 2\theta_p \Wmat_p \Wmat_p^\top) +\frac{\ud\sigma_p^2 (\theta_p)}{\ud\theta_p}(\Imat_N + \theta_p \Wmat_p) (\Imat_N + \theta_p \Wmat_p)^\top, \; \\
&\qquad  p=1,\ldots, P.
\end{split}
\ee 
In the considered setting, $\svec_{p}$ from \eqref{s_m_deriv} is a scalar given by
\be\label{s_m_deriv_scalar}
s_{p} = {\rm{tr}}(\Cmat_p^{-1} \Dmat_p),~p=1, \ldots, P.
\ee
In addition, $\Jmat_{\thetavecsmall_p}$ whose elements are give in \eqref{FIM_compute_theta_sub} is a scalar given by
\be\label{FIM_compute_theta_sub_scalar}
j_{\theta_p}=\frac{1}{2}{\rm{tr}}\left(\Cmat_p^{-1}\Dmat_{p}\Cmat_p^{-1}\Dmat_{p}\right),~p=1,\ldots,P.
\ee 
Consequently, by substituting \eqref{s_m_deriv_scalar} and \eqref{FIM_compute_theta_sub_scalar} into \eqref{zeta_general}, we obtain
\be\label{zeta_define}
\zeta_p \define 2N - \frac{2{\text{tr}}^2 (\Cmat_p^{-1} \Dmat_p)} {{\text{tr}}(\Cmat_p^{-1} \Dmat_p \Cmat_p^{-1} \Dmat_p)}, \quad p = 1, \ldots, P.
\ee
Then, the CRB for this case is obtained by plugging \eqref{kappa_define}, \eqref{C_m_define}, \eqref{S_m_new_theta_depend}, and \eqref{zeta_define} into \eqref{CRB_compute_Omega_block}.\\
\indent
As an example for the parameterized variances, we consider
\be\label{sigma2_theta_depend}
\sigma_p^2 (\theta_p) = \frac{N}{{\text{tr}} ((\Imat_N + \theta_p \Wmat_p) (\Imat_N + \theta_p \Wmat_p)^\top)}
\ee
such that ${\text{tr}}(\Cmat_p)=N, \quad \forall p=1,\ldots,P$. The variance choice from \eqref{sigma2_theta_depend} implies that the IC, $\zvec_p$, is normalized so that its energy ${\rm{E}}\{ \|\zvec_p\|^2 \} ] =N, \quad \forall p=1,\ldots,P$, as mentioned in Section~\ref{sec:BSS}. Substitution of \eqref{sigma2_theta_depend} into \eqref{S_m_new_theta_depend} yields
\begin{equation*}
\begin{split}
\Dmat_p &= \frac{N{\text{tr}}((\Imat_N + \theta_p \Wmat_p) (\Imat_N + \theta_p \Wmat_p)^\top)}{{\text{tr}}^2 ((\Imat_N + \theta_p \Wmat_p) (\Imat_N + \theta_p \Wmat_p)^\top)} (\Wmat_p + \Wmat_p^\top + 2 \theta_p \Wmat_p \Wmat_p^\top) \\
&\quad -  \frac{N{\text{tr}} (\Wmat_p + \Wmat_p^\top + 2 \theta_p \Wmat_p \Wmat_p^\top)}{{\text{tr}}^2 ((\Imat_N + \theta_p \Wmat_p) (\Imat_N + \theta_p \Wmat_p)^\top)} (\Imat_N + \theta_p \Wmat_p) (\Imat_N + \theta_p \Wmat_p)^\top, \\
& \qquad p=1,\ldots, P.
\end{split}
\end{equation*}

\section{Non-Gaussian graph signal BSS}
\label{sec:NGGBSS}
In this section, we consider BSS of ICs which are non-Gaussian graph signals. For separation of such signals, it is then reasonable to exploit both the differences due to the graph correlations as well as non-Gausianity. In general, the BSS estimator in this case can be defined as an optimizer of the following composite objective function
\begin{align}
\label{genObjF}
f(\mathbf{U},\tilde{\mathbf{X}}) \define \lambda \,f_1 (\mathbf{U}, \tilde{\mathbf{X}}) + (1- \lambda) \, f_2 (\mathbf{U}, \tilde{\mathbf{X}})
\end{align}
where $\lambda\in [0,1]$ is a weight parameter that aims to balance between the two parts of the objective function, among which $f_1 (\mathbf{U}, \tilde{\mathbf{X}})$ focuses on the graph correlations and $f_2 (\mathbf{U}, \tilde{\mathbf{X}})$ focuses on non-Gaussianity. We propose two methods that maximize different particular forms of the objective \eqref{genObjF} over $\mathbf{U}$ subject to the orthogonality constraint $\mathbf{U} \mathbf{U}^\top = \mathbf{I}_P$. The identifiability conditions for the methods are then also discussed.

\subsection{Graph JADE}
\label{sec:GraphJADE}
First, we suggest a method that we call {\it Graph JADE}. the name comes from the fact that the corresponding objective function to be minimized is selected to be
\begin{align}
\label{GraphJADE}
f_{\rm GraphJADE} &= \lambda \,f_1 (\mathbf{U}, \tilde{\mathbf{X}}) + (1- \lambda) \sum_{k=1}^P\sum_{l=1}^P\|\diag(\mathbf{U}\hat{\mathbf{C}}^{k,l}\mathbf{U}^\top)\|^2 
\end{align}
where $f_1 (\mathbf{U}, \tilde{\mathbf{X}})$ is the graph decorrelation objective defined in \eqref{GraDeobj} and the second term $f_2 (\mathbf{U}, \tilde{\mathbf{X}}) = \sum_{k=1}^P \sum_{l=1}^P \| \diag( \mathbf{U} \hat{\mathbf{C}}^{k,l} \mathbf{U}^\top) \|^2$ is the objective function of JADE that uses the following set of fourth order cumulant matrices 
\begin{align}
\label{JADEmat}
\hat{\mathbf{C}}^{k,l} =& \sum_{i=1}^N \tilde{x}_{k,i} \tilde{x}_{l,i} \tilde{\bf x}_{i} \tilde{\bf x}_{i}^\top - \mathbf{E}_{P \times P}^{k,l} - \mathbf{E}_{P \times P}^{l,k} - \mbox{tr} (\mathbf{E}_{P \times P}^{k,l} )\mathbf{I}_P, \quad k,l=1,\dots,P. 
\end{align}

Fortunately, the maximization of \eqref{GraphJADE} subject to the orthogonality constraint appears to be a straightforward extension of the well-known joint approximate diagonalization algorithm. Indeed, since the maximization of the graph decorrelation part of the composite objective \eqref{GraphJADE}, that is, $f_1 (\mathbf{U}, \tilde{\mathbf{X}})$ as well as the maximization of the JADE part -- $f_2 (\mathbf{U}, \tilde{\mathbf{X}})$, can be done in terms of joint approximate diagonalization, the maximization of the sum can be done by joint approximate diagonalization of all matrices involved. 
It can be implemented, for example, by the algorithm in~\cite{Clarkson1988} based on Givens rotations, and it yields a fast solution if the number of ICs, $P$, is not too large. 

In general, the selection of the weight parameter $\lambda$ depends on which difference between ICs is stronger: the graph dependence or the deviation of the ICs' distribution from the Gaussian one. However, we confirmed by extensive simulations that one value of $\lambda$ can succeed in a wide range of setups and there is no apparent need to find the case-specific optimal value of $\lambda$. For Graph JADE, we suggest\footnote{\label{lambdachoice} A rule of thumb for choosing $\lambda$ is that the expected values of $\lambda f_1$ and $(1-\lambda ) f_2$ should be approximately equal, when the ICs are non-Gaussian graph signals.} $\lambda = 0.8$. 

\subsection{Graph FastICA}
Another alternative is the FastICA based graph decorrelation, called hereafter as {\it Graph FastICA}. The corresponding optimization problem is then to find the orthogonal matrix $\mathbf{U}=(\mathbf{u}_1, \cdots, \mathbf{u}_P)^\top$ that maximizes the following composite objective function
\begin{align}
\label{GraphFastICA}
f_{\rm GraphFastICA} &= \lambda  f_1 (\mathbf{U}, \tilde{\mathbf{X}}) + (1 - \lambda) \sum_{j=1}^P\left(\sum_{i=1}^N \frac{1}{N}G(\boldsymbol{u}_j^\top \tilde{\mathbf{x}}_{i})\right)^2
\end{align}
where $f_1 (\mathbf{U}, \tilde{\mathbf{X}})$ is again given by \eqref{GraDeobj} and the second term in the composite objective function is that of the squared symmetric FastICA~\cite{Miettinenetal2017}. Notice that in the classical symmetric FastICA, the absolute value replaces the second power. The squared version puts more weight to the components that are more non-Gaussian, which is an advantage when the FastICA objective is extended with other objectives such as in here. The function $G (\cdot)$ is a twice continuously differentiable nonlinear and nonquadratic function satisfying $\E \{G(y)\} = 0$ for a standardized Gaussian random variable $y$. We use arguably the most popular choice of function $G(x) = \mbox{log} (\mbox{cosh}(x)) - \E \{\mbox{log}(\mbox{cosh}(y))\}$, which is referred to as {\em tanh} according to the first derivative $g(x)=G'(x) = \mbox{tanh}(x)$. Choice%\footnotemark[\ref{lambdachoice}] 
of $\lambda$ here depends on the function $G$.

To address the problem of maximizing \eqref{GraphFastICA} subject to $\mathbf{U}\mathbf{U}^\top =\mathbf{I}_P$, we first consider the maximization of $f_1$. The corresponding nonconvex optimization problem is 
\begin{equation} \label{optF1}
\underset{\mathbf{U}\mathbf{U}^\top =\mathbf{I}_{P}}{\max} \sum_{p=1}^P \sum_{k=1}^K \| \diag( \mathbf{U} \tilde{\mathbf{S}}_p^k \mathbf{U}^\top ) \|^2,
\end{equation}
where for the sake of brevity we use a simplified notation $\tilde{\mathbf{S}}_p^k$ for  $\tilde{\mathbf{S}} (\tilde{\mathbf{X}}, \mathbf{W}_p^k)$. 

Since problem \eqref{optF1} is nonconvex, we can find only a suboptimal solution, which is a stable point, i.e., the point at which Karush-Kuhn-Tucker (KKT) conditions are satisfied. Towards this end, we formulated the corresponding  Lagrangian function as
\begin{align*} 
L_g(\mathbf{U}, \mathbf{\Xi}) =& \sum_{i=1}^P \sum_{p=1}^P \sum_{k=1}^K (\mathbf{u}_i^\top \tilde{\mathbf{S}}_p^k \mathbf{u}_i)^2 - \sum_{i=1}^P \xi_{i,i} (\mathbf{u}_i^\top \mathbf{u}_i-1) - \sum_{i=1}^P \sum_{j=1}^{i-1} \xi_{j,i} \mathbf{u}_j^\top \mathbf{u}_i
\end{align*}
where $\mathbf{\Xi}$ is a symmetric matrix containing the Lagrange multipliers $\xi_{j,i}, \quad i,j = 1, \cdots, P$. One KKT condition is obtained by taking the first derivative of $L_g(\mathbf{U}, \mathbf{\Xi})$ with respect to $\mathbf{u}_j$ and setting it to zero. As a result, the following equation is obtained 
\begin{align*}
2\sum_{p=1}^P \sum_{k=1}^K \tilde{\mathbf{S}}_p^k \mathbf{u}_j \mathbf{u}_j^\top \tilde{\mathbf{S}}_p^k \mathbf{u}_j 
= 2\xi_{j,j}\mathbf{u}_j+\sum_{i=1}^{j-1} \xi_{i,j} \mathbf{u}_i+\sum_{i=j+1}^{P} \xi_{j,i} \mathbf{u}_i.
\end{align*}
When both sides of the above equation are multiplied from left by $\mathbf{u}_i$, $i\neq j$, the orthogonality of $\mathbf{U}$ yields
\begin{align*}
&2\mathbf{u}_i^\top\sum_{p=1}^P \sum_{k=1}^K \tilde{\mathbf{S}}_p^k \mathbf{u}_j \mathbf{u}_j^\top \tilde{\mathbf{S}}_p^k \mathbf{u}_j = \xi_{i,j}\ (\xi_{j,i})\text{ if } i<j\ (i>j).
\end{align*}

Since $\mathbf{\Xi}$ is a symmetric matrix, the following equations\footnote{See also \cite{miettinen2016separation} for the case of time series autocovariance matrices.} have to hold true
\begin{align}
\label{EstEq}
&\mathbf{u}_i^\top \sum_{p=1}^P \sum_{k=1}^K \tilde{\mathbf{S}}_p^k \mathbf{u}_j \mathbf{u}_j^\top \tilde{\mathbf{S}}_p^k \mathbf{u}_j = \mathbf{u}_j^\top \sum_{p=1}^P \sum_{k=1}^K \tilde{\mathbf{S}}_p^k \mathbf{u}_i \mathbf{u}_i^\top \tilde{\mathbf{S}}_p^k \mathbf{u}_i \notag \\
& \qquad \mathbf{u}_i^\top \mathbf{u}_i = 1, \quad i = 1, \cdots, P \\
& \qquad \mathbf{u}_i^\top \mathbf{u}_j = 0, \quad i\neq j, \quad i,j = 1, \cdots, P .\notag
\end{align}
Thus, it is also guaranteed that the other KKT condition, which is the constraint of problem \eqref{optF1} is also satisfied. Equations \eqref{EstEq} can be then used for estimating the rows of $\mathbf{U}$. 

Combining \eqref{EstEq} with the FastICA update rule \cite{HyvarinenOja1997}, which applies for the second term of the composite objective function \eqref{GraphFastICA} straightforwardly, a fixed-point algorithm for updating $\mathbf{u}_j,\ j=1,\dots,P$ of our proposed Graph FastICA algorithm can be then constructed as
\begin{align}
&\mathbf{u}_j  \leftarrow 2 \lambda \sum_{p=1}^P \sum_{k=1}^K \tilde{\mathbf{S}}_p^k \mathbf{u}_j \mathbf{u}_j^\top \tilde{\mathbf{S}}_p^k \mathbf{u}_j +(1-\lambda)\left( \frac{1}{N} \sum_{i=1}^N G(\mathbf{u}_j^\top \tilde{\mathbf{x}}_{i}) \right) \notag \\
&\times  \left( \frac{1}{N} \sum_{i=1}^N \tilde{\mathbf{x}}_{i} g(\mathbf{u}_j^\top \tilde{\mathbf{x}}_{i}) - \frac{1}{N} \sum_{i=1}^N g'(\mathbf{u}_j^\top \tilde{\mathbf{x}}_{i}) \mathbf{u}_j \right).
\label{alg}
\end{align}  
To be specific, the second part of the update rule \eqref{alg} is the updating step of the squared symmetric FastICA \cite{Hyvarinen1999}. Here $g'(\cdot)$ is the first-order derivative function of $g(\cdot)$, or equivalently, the second-order derivative function of $G(\cdot)$.

After updating all rows of $\mathbf{U}=(\mathbf{u}_1, \cdots, \mathbf{u}_P)^\top$, the orthogonalization has to be also ensured. This orthogoanlization can be easily performed by updating $\mathbf{U}$ as  $\mathbf{U} \leftarrow \left(\mathbf{U}\mathbf{U}^\top\right)^{-1/2}\mathbf{U}$. The above update and orthogonalization steps are repeated until convergence. To ensure that the vectors of the two parts (top row and bottom row) in~\eqref{alg} do not point to opposite directions, a vector is multiplied by $-1$, if in that vector the element, which has the largest absolute value, is negative. 

\subsection{Identifiability Conditions}
The identifiability conditions for the matrix of ICs $\mathbf{Z}$, which the proposed methods require for consistent estimation of the unmixing matrix, can be derived by studying the two parts, graph decorrelation and non-Gaussianity based separation, separately. Since graph decorrelation uses joint approximate diagonalization of graph autocorrelation matrices, it is required that for each pair of ICs $\mathbf{z}_i$ and $\mathbf{z}_j$, there is at least one matrix $\mathbf{W}\in\mathcal{W}$ such that the corresponding graph autocorrelations of $i$th and $j$th  ICs differ, i.e., the $i$th and $j$th diagonal elements of $\tilde{\mathbf{S}}(\mathbf{Z},\mathbf{W})$ differ. This condition follows from the fact that for the diagonal matrix with equal diagonal elements, $\mathbf{B} = b \mathbf{I}_{P}$, the matrix $\mathbf{V} \mathbf{B} \mathbf{V}^\top$ is diagonal for any orthogonal matrix $\mathbf{V}$. Whereas when there are unequal diagonal elements for all pairs of ICs, the orthogonal matrix, which makes the graph autocorrelation matrices in the set $\{ \tilde{\mathbf{S}} (\tilde{\mathbf{X}}, \mathbf{W}_p^k) \}_{p=1,\dots,P,\ k=1,\dots,K}$ as diagonal as possible, is unique up to sign changes of its rows. 

For an ICA method, which is based on non-Gaussianity, the well-known principle is that there can be at most one component which is similar to Gaussian component with respect to the measure of non-Gaussianity specific for the method. In JADE, the measure is the fourth moment, while in FastICA, it is a function $G$. Hence, the identifiability conditions for (a)~Graph JADE and (b)~Graph FastICA using the set of adjacency matrices and their powers, $\mathcal{W}$, can be stated as follows. 

For any pair of ICs $\boldsymbol z_i$ and $\boldsymbol z_j$ (i)~there is a matrix $\mathbf{W}\in \mathcal{W}$ such that $\E\{\diag(\tilde{\mathbf{S}}(\mathbf{Z}, \mathbf{W}))_i\}\neq \E\{\diag(\tilde{\mathbf{S}}(\mathbf{Z}, \mathbf{W}))_j\}$ or (ii)~(a) 
$\E\{N^{-1}\sum_{n=1}^N z_{i,n}^4\}\ne 3$ or $\E\{N^{-1}\sum_{n=1}^N z_{j,n}^4\}\ne 3$; (b) $\E\{N^{-1}$ $\sum_{n=1}^N G(z_{i,n})\}\ne 0$ or $\E\{N^{-1}\sum_{n=1}^N G(z_{j,n})\}\ne 0$. %Notice that requiring only part (i) or (ii) of the above conditions relaxes the assumptions needed for separation of the ICs by graph decorrelation only, where (i) is needed, or by ICA methods only, where (ii) is needed.

Allowing several adjacency matrices in $\mathcal{W}$ does not expand the identifiability conditions as compared to the case when there is only one adjacency matrix, provided that this matrix is chosen correctly. However, in Section~\ref{sec:simul}, it is shown that the gain in separation efficiency can be significant.

\section{Numerical Study}
\label{sec:simul}

The accuracy of the unmixing matrix estimation is measured by the minimum distance (MD) index~\cite{Ilmonenetal2010a} 
\[D(\hat{\mathbf{\Gamma}}) = \frac{1}{\sqrt{P-1}}\inf_{\mathbf{C}\in\mathcal{C}}\|\mathbf{C}\hat{\mathbf{\Gamma}}\mathbf{\Omega}-\mathbf{I}_p\|
\]
where $\mathcal{C} = \{ \mathbf{C}\ : \mbox{ each row and column of $\mathbf{C}$ has exactly one}$ $\mbox{non-zero element} \}$.
The MD index is zero when $\hat{\mathbf{\Gamma}}$ can be obtained from $\mathbf{\Omega}^{-1}$ by permutation and rescaling of the rows. The maximum value one occurs when $\hat{\mathbf{\Gamma}}$ has rank one.
When $\mathbf{\Omega}=\mathbf{I}_{P\times P}$ and the estimator $\hat{\mathbf{\Gamma}}$ is asymptotically normal (the elements $\hat{{\gamma}}_{ij}$ are asymptotically normal), which is often the case (see for example \cite{Miettinenetal2017, Ilmonenetal2010b, Miettinenetal2012}), the minimum distance index is related to the sum of variances of the off-diagonal elements of $\hat{\mathbf{\Gamma}}$ as 
$N(P-1)\E\{D(\hat{\mathbf{\Gamma}})^2\}\to \sum_{i\neq j=1}^P\mbox{var}\{\hat{{\gamma}}_{ij} \} ,\ \mathrm{as}\  N\to \infty$.
\begin{figure}%[htb]
\begin{minipage}[b]{1.0\linewidth}
\centering
\includegraphics[width=15.8cm]{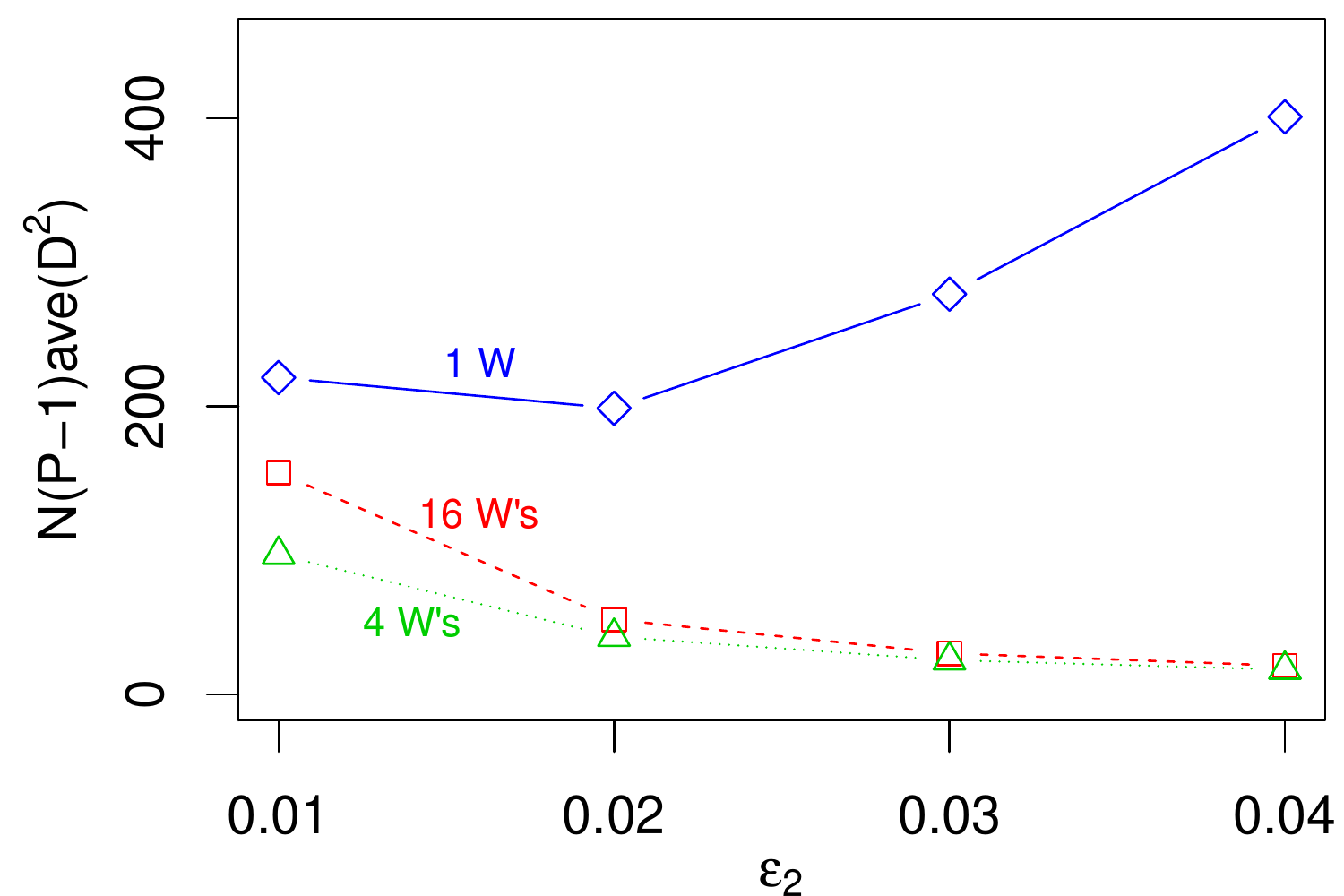}	
\end{minipage}
\caption{The average values of $N(P-1)D^2(\hat{\mathbf{\Gamma}})$ over 1000 repetitions.}
\label{fig1}
\end{figure}

\subsection{Different Adjacency Matrices}

We begin with studying the use of several adjacency matrices when the ICs have different dependence structures. We consider four models with increasing degree of difference between the adjacency matrices. In each model, there are $P=4$ ICs which are Gaussian GMA$(1)$ signals of size $N=500$ and the GMA coefficients are $\theta=0.32,0.16,0.08,0.04$. The adjacency matrices are generated so that the first one $\mathbf{W}_1 = \mathbf{\Delta}_{\epsilon}$ is obtained from Erd\"{o}s--R\'{e}nyi random graph model \cite{ErdosRenyi1959, Gilbert1959}, where $\epsilon$ is a constant probability of edge existence between two randomly picked nodes, and $\epsilon=0.05$. The other three adjacency matrices $\mathbf{W}_2, \mathbf{W}_3, \mathbf{W}_4$ are obtained from $\mathbf{W}_1$ using the graph error model \cite{MiettinenVorobyovOllila2018, USjournal}
$\mathbf{W}_{\epsilon_1,\epsilon_2} = \mathbf{W} - \mathbf{\Delta}_{\epsilon_1}\odot \mathbf{W} + \mathbf{\Delta}_{\epsilon_2}\odot (\mathbf{1}_{N\times N}-\mathbf{W})$,
where $\epsilon_1$ and $\epsilon_2$ are the probabilities of removing and adding an edge from/to a given graph $\mathbf{W}$. In the first model parameter pair $(\epsilon_1,\epsilon_2)=(0.19,0.01)$ is used, in the second $(0.38,0.02)$, in the third  $(0.57,0.03)$, and in the fourth $(0.76,0.04)$. The values of $\epsilon_1$ and $\epsilon_2$ are coordinated so that the graphs have approximately equally many edges. In each model, adjacency matrices $\mathbf{W}_2, \mathbf{W}_3, \mathbf{W}_4$ are generated using the same pair $\epsilon_1, \epsilon_2$. Three graph decorrelation based estimators ($\lambda=1$) are compared, the first one $\mathbf{\Gamma}_1$ uses only $\mathbf{W}_1$, the second one $\mathbf{\Gamma}_2$ uses $\{\mathbf{W}_1, \mathbf{W}_2, \mathbf{W}_3, \mathbf{W}_4\}$, and the third one $\mathbf{\Gamma}_3$ uses also extra twelve Erd\"{o}s--R\'{e}nyi adjacency matrices in addition to the four above which are independent of $\mathbf{W}_1$. For each estimator in each model, the average of $N(P-1)D(\hat{\mathbf{\Gamma}})^2$ over 1000 repetitions is plotted in Fig.~\ref{fig1}. It can be seen that $\mathbf{\Gamma}_2$ is better than $\mathbf{\Gamma}_1$ as expected and the difference grows with $\epsilon_2$. This is because $\mathbf{\Gamma}_1$ becomes worse when the adjacency matrices differ more, and also because $\mathbf{\Gamma}_2$ improves. The extra adjacency matrices in $\mathbf{\Gamma}_3$ seem to hamper the estimation, but not very significantly. %when it is difficult in the first place. For third and fourth models, the difference is negligible. 

\subsection{ML estimation and CRB}

In this subsection, we first define an ML estimator for the BSS problem with GMA source signals, and then compare it (as well as  GraDe) to the corresponding CRB derived in Section~\ref{sec:CRB} through numerical examples. Unlike GraDe, the ML estimator assumes a model for the source signal, and thus, we can expect it to perform better than GraDe when the model is correctly specified. Besides the assumption on the source model, the use of ML estimator in practice is limited to graphs of moderate size because of computational issues. However, it is interesting to see for the theoretical study how close the ML-based solution and the solution based on the objective \eqref{GraDeobj} are to the CRB. For simplicity, we consider in what follows the BSS model with $P=2$ ICs and GMA signals of order 1, which will be the setup in our numerical studies as well.

Assume that the ICs $\mathbf{z}_1$ and $\mathbf{z}_2$ are Gaussian GMA$(1)$ signals with adjacency matrices $\mathbf{W}_1$ and $\mathbf{W}_2$, GMA coefficients $\theta_1$ and $\theta_2$, and variance parameters $\sigma_1^2$ and $\sigma_2^2$, respectively. The log-likelihood (after removing constants) as a function of the graph signal $\mathbf{z}$, GMA coefficient $\theta$, and adjacency matrix $\mathbf{W}$, is given as
\begin{align}
\label{logLik}
L(\mathbf{z},\theta,\mathbf{W}) &= -0.5\mathbf{z}^\top\mathbf{C}^{-1} (\theta,\sigma^2,\mathbf{W}) \mathbf{z}  -\log(\det(\mathbf{C} (\theta,\sigma^2,\mathbf{W}))) 
\end{align}
where the covariance matrix of the graph signal $\mathbf{z}$ is 
$\mathbf{C} (\theta, \sigma^2, \mathbf{W}) = \sigma^2 (\mathbf{I}_N + \theta\mathbf{W}) (\mathbf{I}_N + \theta\mathbf{W})^\top$.

The basic algorithm for the ML estimator of the unmixing matrix contains the following steps.
\begin{itemize}
	\item[1.] Find an initial estimate $\hat{\mathbf{Z}}_0 = (\hat{\mathbf{z}}_{01}, \hat{\mathbf{z}}_{02})^\top = \hat{\mathbf{U}}_0 \tilde{\mathbf{X}}$ using the objective \eqref{GraDeobj} for obtaining $\hat{\mathbf{U}}_0$.
	
	\item[2.] Estimate the GMA coefficients $\hat{\theta}_1$ and $\hat{\theta}_2$ of $\hat{\mathbf{z}}_{01}$ and $\hat{\mathbf{z}}_{02}$, respectively, if they are also unknown.
	
	\item[3.] Find the orthogonal matrix $\hat{\mathbf{U}}_1$ such that $(\hat{\mathbf{z}}_{11},\hat{\mathbf{z}}_{12})^\top=\hat{\mathbf{U}}_1\tilde{\mathbf{Z}}_0$ maximizes $L(\hat{\mathbf{z}}_{11},\hat{\theta}_1,\mathbf{W}_1)+L(\hat{\mathbf{z}}_{12},\hat{\theta}_2,\mathbf{W}_2)$.
	
	\item[4.] Compute the final estimate of $\mathbf{\Gamma}$ as  $\hat{\mathbf{U}}_1\hat{\mathbf{U}}_0\hat{\mathbf{S}_0}^{-1/2}$.
\end{itemize}

In steps 2 and 3 of the above algorithm, the maximization is carried out by grid search. Notice that the $2\times 2$ orthogonal matrix $\mathbf{U}$ can be represented using a single parameter $\phi$ as 
\[
\mathbf{U}=\left[ {\begin{array}{cc}
	\text{cos}(\phi) & -\text{sin}(\phi)  \\
	\text{sin}(\phi) & \text{cos}(\phi)  \\					
	\end{array} } \right].
\]
A $P\times P$ orthogonal matrix has $P(P-1)/2$ free parameters, and therefore, the estimation in this way becomes quickly infeasible. For $P>2$, the hybrid exact-approximate diagonalization (HEAD) approach \cite{YEREDOR_GAUSSIAN_CRB} would give an approximation of the ML estimate.

A natural idea for improving the above algorithm is to iterate steps 2 and 3 several times until convergence. However, in our simulations, this idea did not lead to any significant improvements of the results, and thus we will use just the one step version in the simulation studies of Section~\ref{sec:simul}.

To simplify the calculations, we note that $\sigma^2_p$, $p=1,2$, do not need to be estimated by a grid search simultaneously with $\theta_p$,  $p=1,2$, but since $\mathbf{z}_p$ can be assumed without loss of generality in BSS model to be standardized, $\sigma^2_p$ can be connected to $\theta_p$ by setting $\sigma^2_p(\theta_p) = N/\text{tr} ((\Imat_N + \theta_p \Wmat_p) (\Imat_N + \theta_p \Wmat_p)^\top)$ as in \eqref{sigma2_theta_depend}.

Simulations are carried out in BSS models with two ICs which are built from three graphs generated by stochastic block model with two communities, geometric model, and Erd\"{o}s--R\'{e}nyi model. The size of each graph is $N=250$. The community graph is created by dividing the nodes in two blocks and connecting nodes within each block with probability 0.13 and between blocks with probability 0.01. The geometric graph is generated by scattering vertices randomly on the unit square and connecting a pair of vertices if their Euclidean distance is smaller than 0.16. In Erd\"{o}s--R\'{e}nyi model the parameter (probability of edge existence between a pair of nodes) is 0.07. The numbers above are chosen so that the number of edges in the graphs would be approximately equal. CRB calculations and estimator performance simulations are performed for all combinations, but the results are presented in the following four combinations of these graphs: (C1) first IC is based on community graph and second IC on geometric graph, (C2) first IC is based on community and second on Erd\"{o}s--R\'{e}nyi graph, (C3) first IC is based on geometric and second on Erd\"{o}s--R\'{e}nyi graph, and (C4) both ICs are based on the same community graph. The GMA$(1)$ parameter of the first IC is kept fixed at $\theta_1 = 0.1$, and the GMA$(1)$ parameter $\theta_2$ of the second IC is allowed to vary in the range $0.01,0.02,\dots,0.4$. The results for community and Erd\"{o}s--R\'{e}nyi graphs are very similar, so for example (C1) is almost equivalent to the case where first IC is based on Erd\"{o}s--R\'{e}nyi graph and second IC on geometric graph. The geometric graph is more challenging than the other two for all algorithms, and the CRB confirms it as well.  

In Fig.~\ref{fig2}, $\text{tr}(\textbf{CRB}_{\Imat_2})$ and $((\textbf{CRB}_{\Imat_2})_{2,2} + (\textbf{CRB}_{\Imat_2})_{3,3})$ are compared to the sum of variances of all the elements of $\hat{\mathbf{\Omega}}$ and the sum of off-diagonal elements of $\hat{\mathbf{\Omega}}$ in 2000 runs of each model. Both CRB values and variances are multiplied by $N$. We include three estimators, the proposed Gaussian graph signal BSS method, referred hereafter simply as GraDe, the proposed ML estimator where an initial estimate is obtained by using GraDe, and an oracle estimator where the initial estimate is set to the true value. 

\begin{figure*}[]
	\begin{minipage}[b]{1.0\linewidth}
		\centering
		\includegraphics[width=7.0cm]{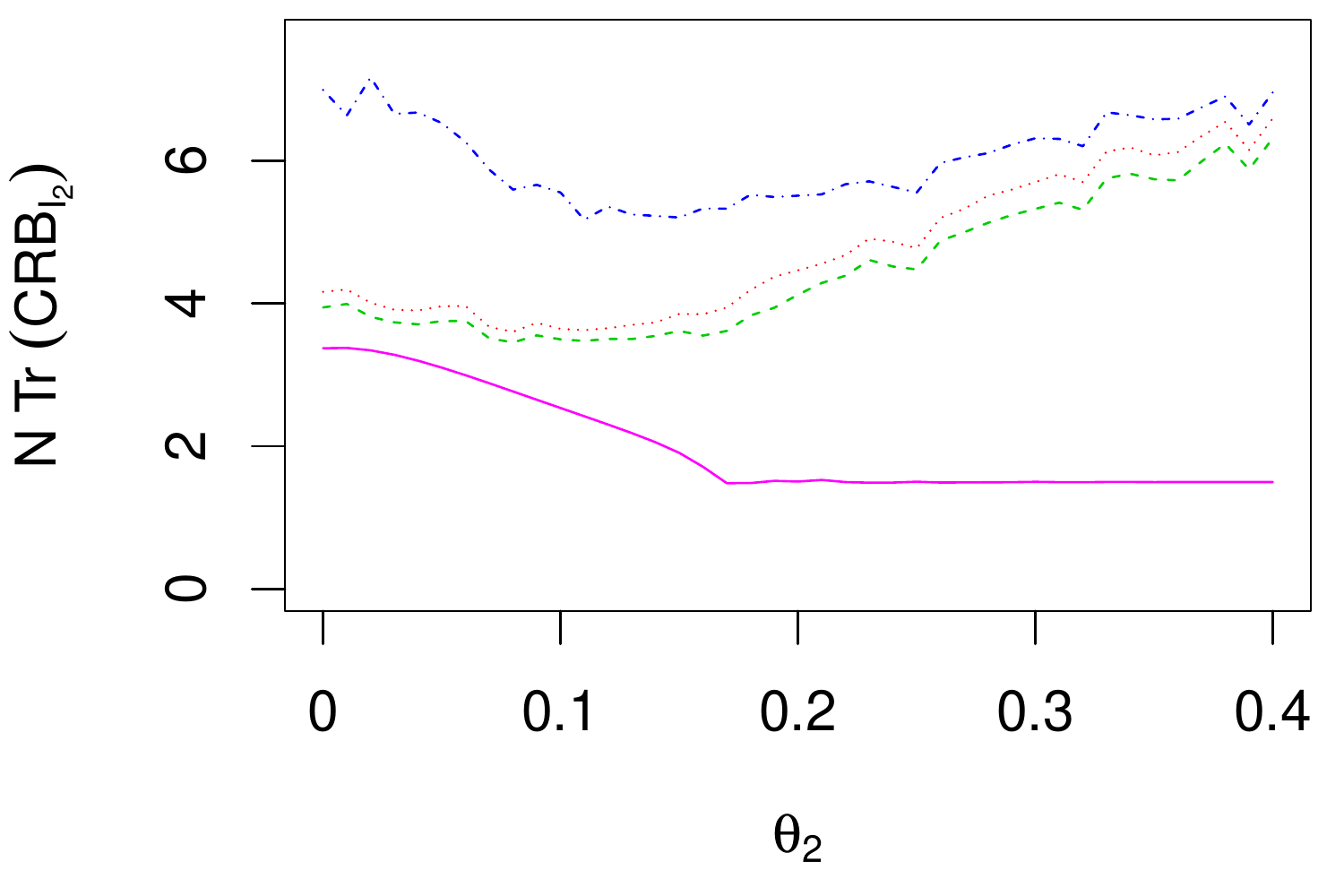}
		\hspace{1.1cm}
		\includegraphics[width=7.0cm]{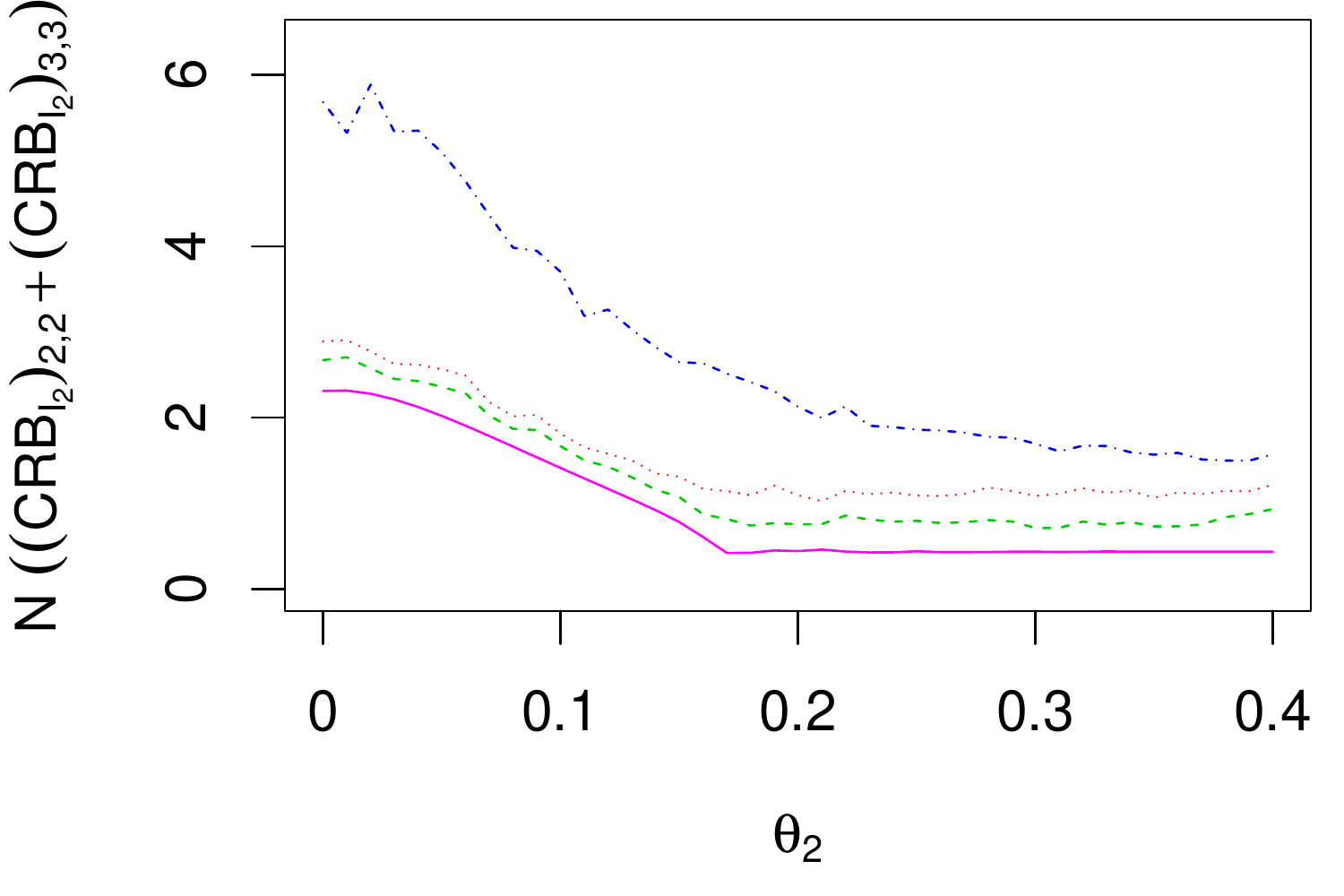}
		\\
		\vspace{0.2cm}
		\includegraphics[width=7.0cm]{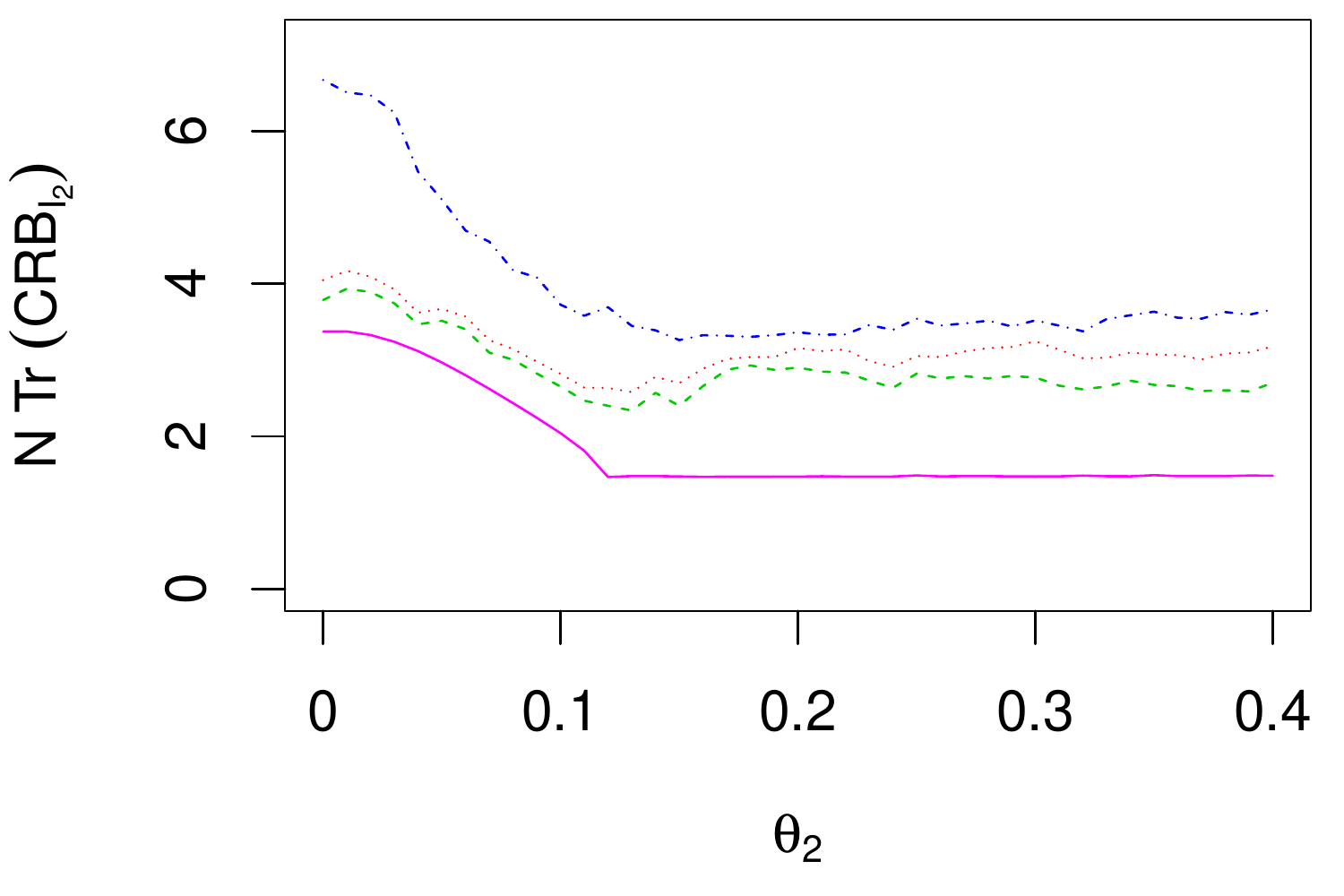}
		\hspace{1.1cm}
		\includegraphics[width=7.0cm]{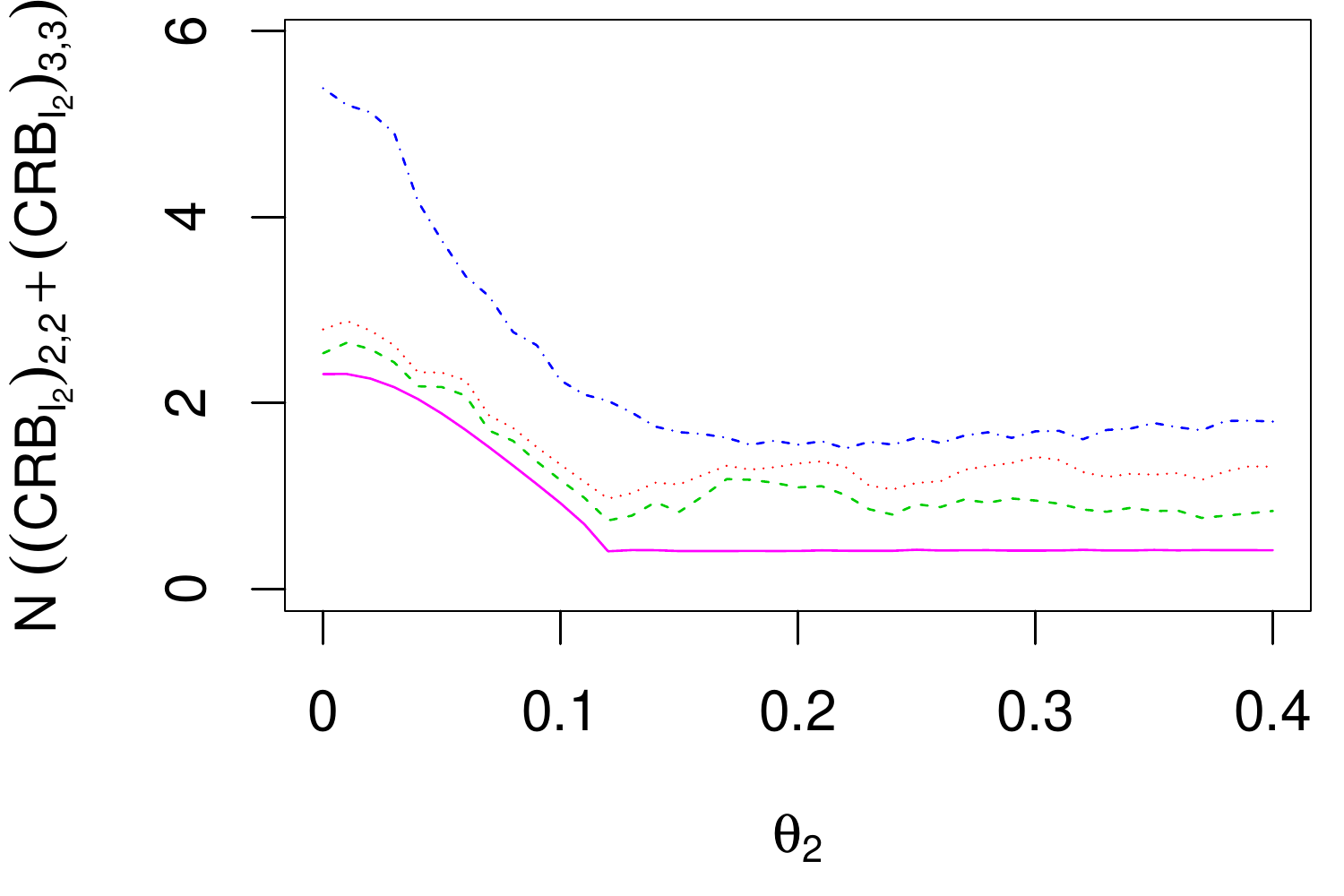} \\
		\vspace{0.2cm}
		\includegraphics[width=7.0cm]{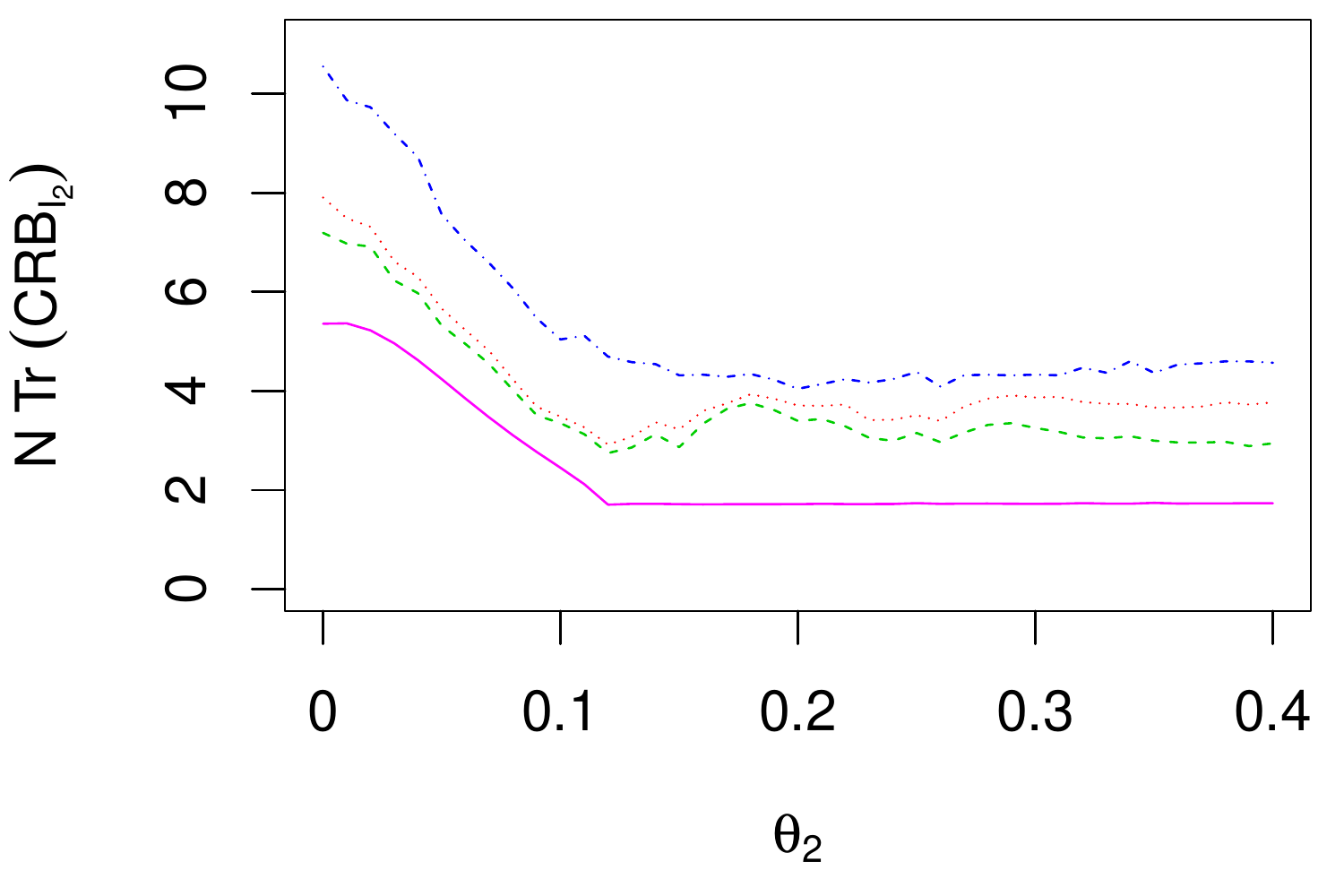}
		\hspace{1.1cm}
		\includegraphics[width=7.0cm]{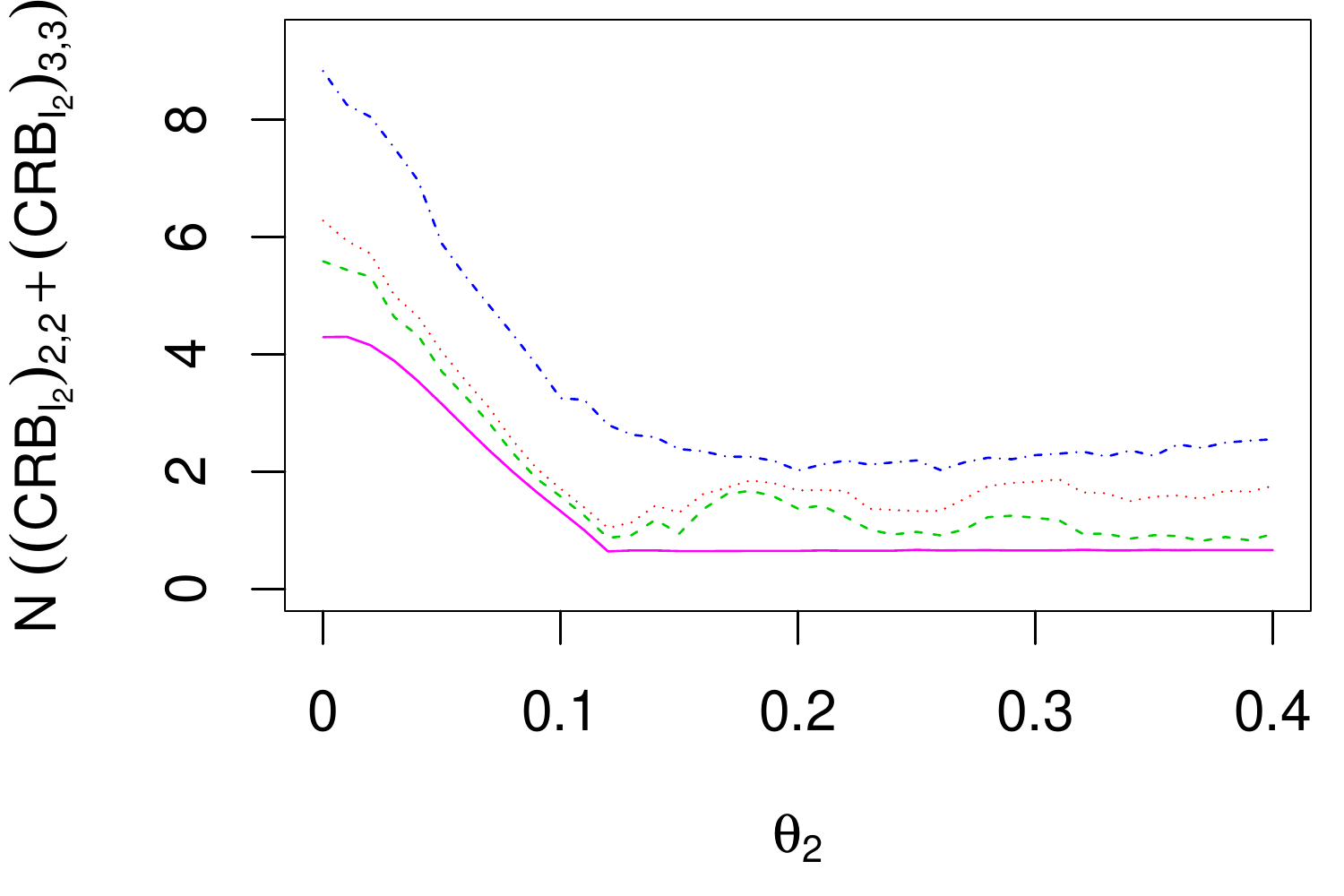} \\
		%\hspace{-.1cm}
		\includegraphics[width=7.0cm]{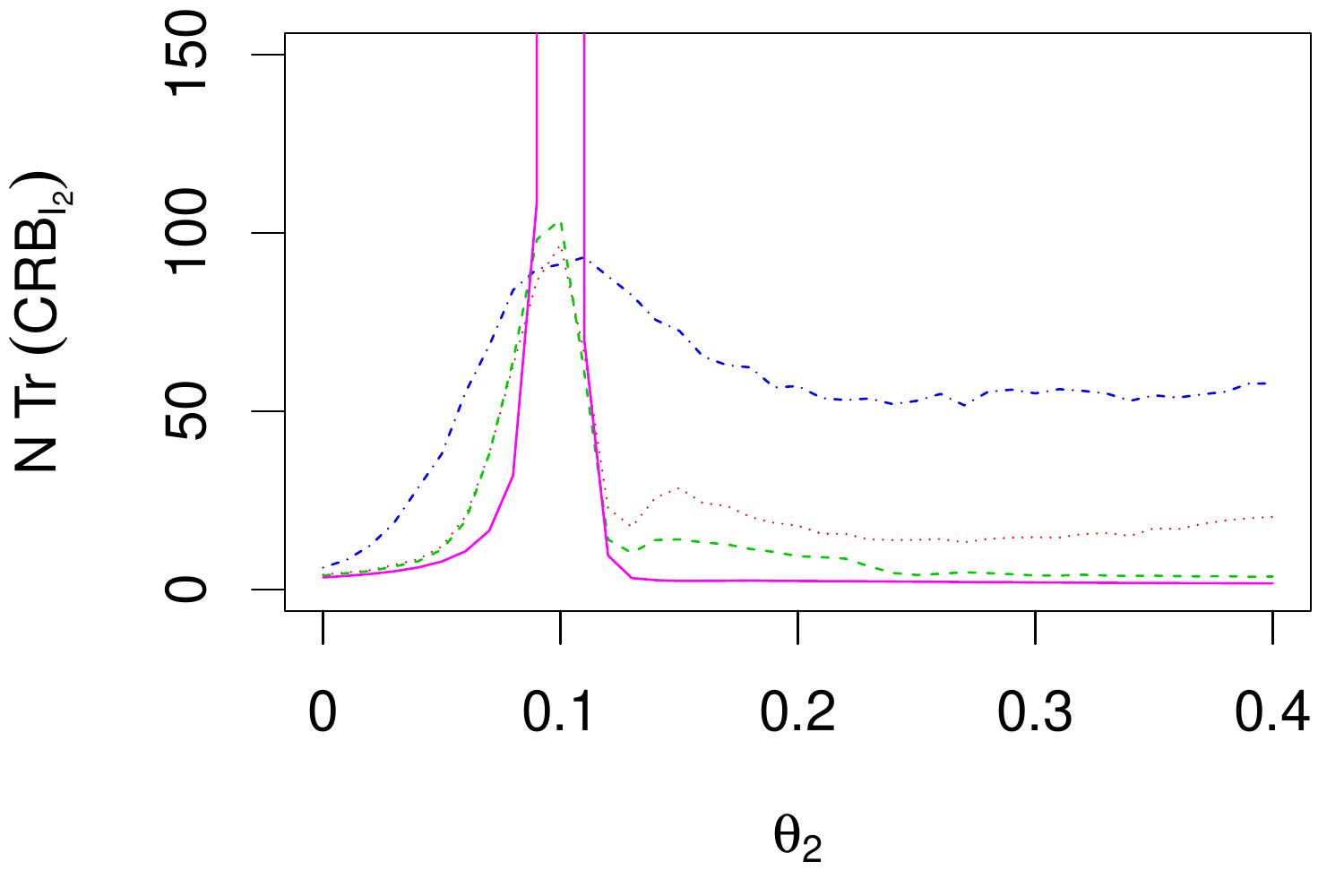}
		\hspace{1.1cm}
		\includegraphics[width=7.0cm]{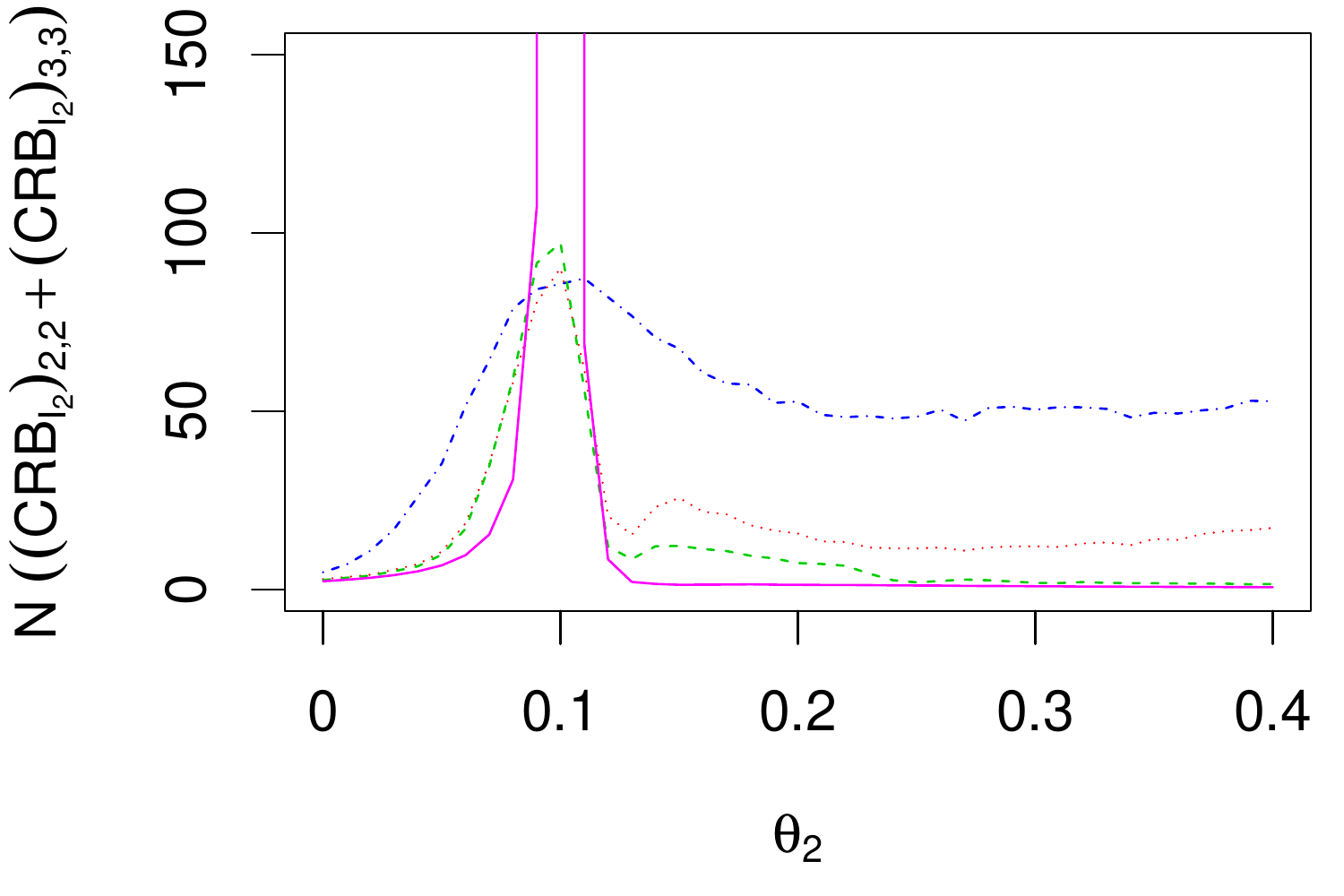}\\
		\includegraphics[width=19.5cm]{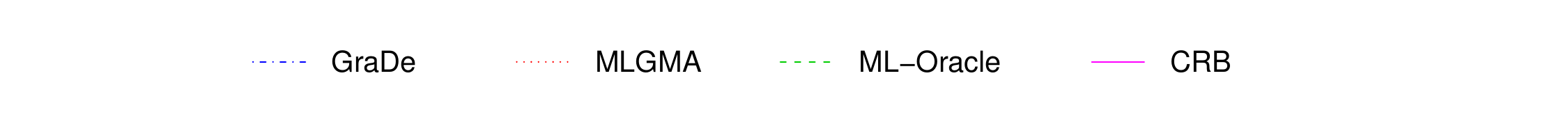}
	\end{minipage}
	\vspace{-1.0cm}
	\caption{$\text{Tr}(\textbf{CRB}_{\Imat_2})$ on the left-hand side and  $((\textbf{CRB}_{\Imat_2})_{2,2}+(\textbf{CRB}_{\Imat_2})_{3,3})$ on the rigth-hand side and the corresponding simulation results in 2000 runs. From top to bottom: (C1) Two-communities and distance based graphs, (C2) two-communities and Erd\"{o}s-R\'{e}nyi graphs, (C3) distance based and Erd\"{o}s-R\'{e}nyi graphs and (C4) the same two-communities graph for both ICs.}
	\label{fig2}
\end{figure*}

\begin{figure*}[]
	\begin{minipage}[b]{1.0\linewidth}
		\centering
		\includegraphics[width=6.1cm]{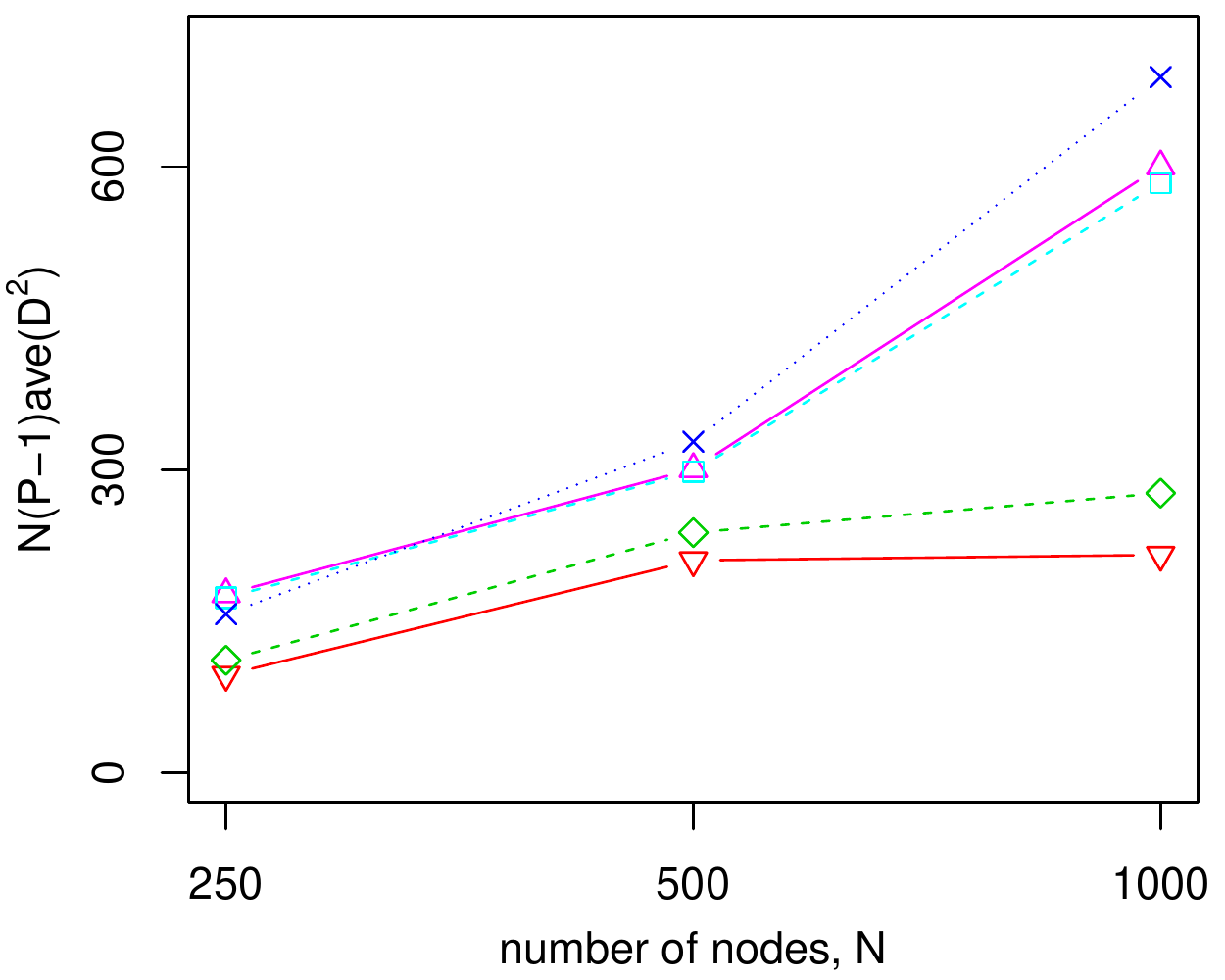}
		\hspace{1.1cm}
		\includegraphics[width=5.7cm]{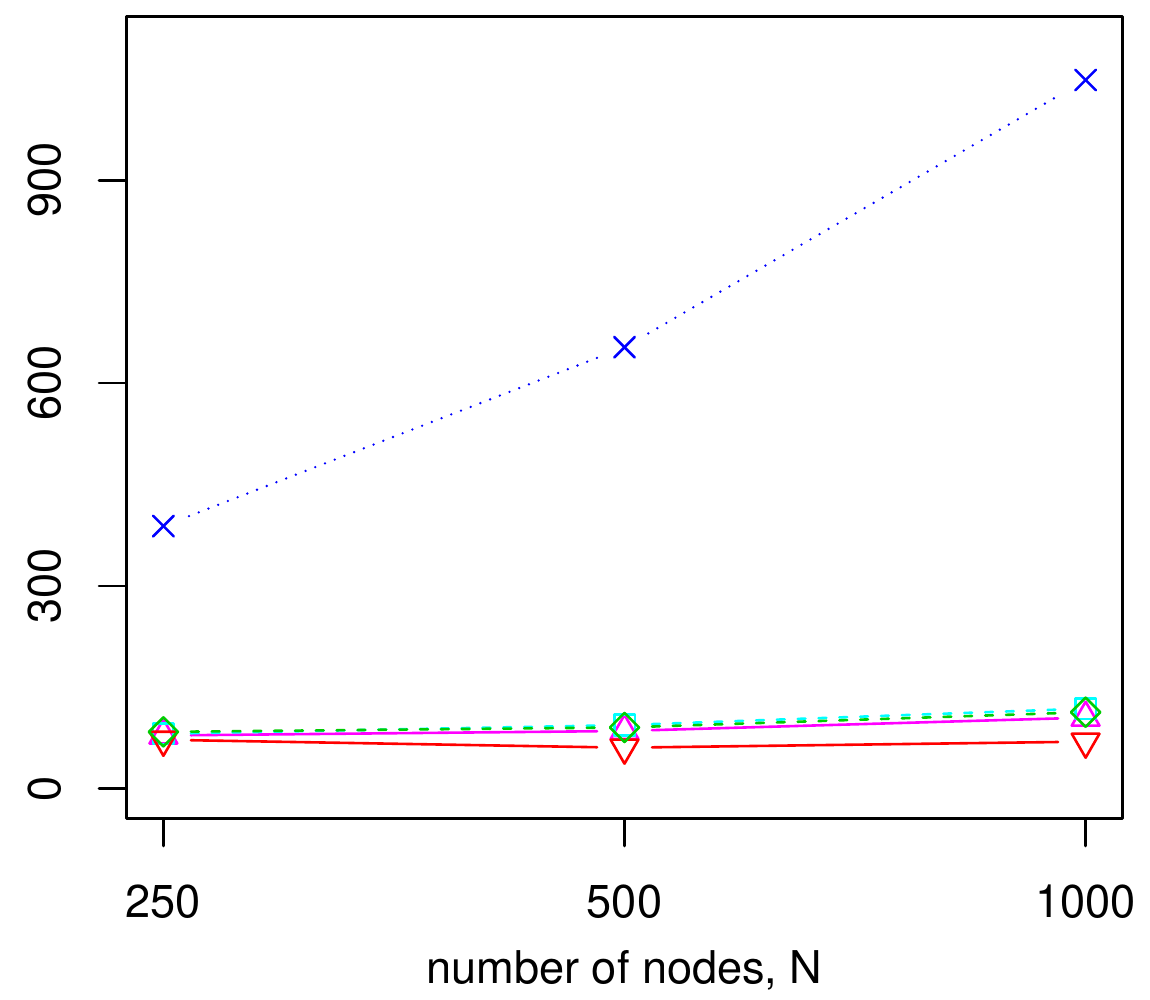}
		\\
		\vspace{0.2cm}
		\includegraphics[width=6.1cm]{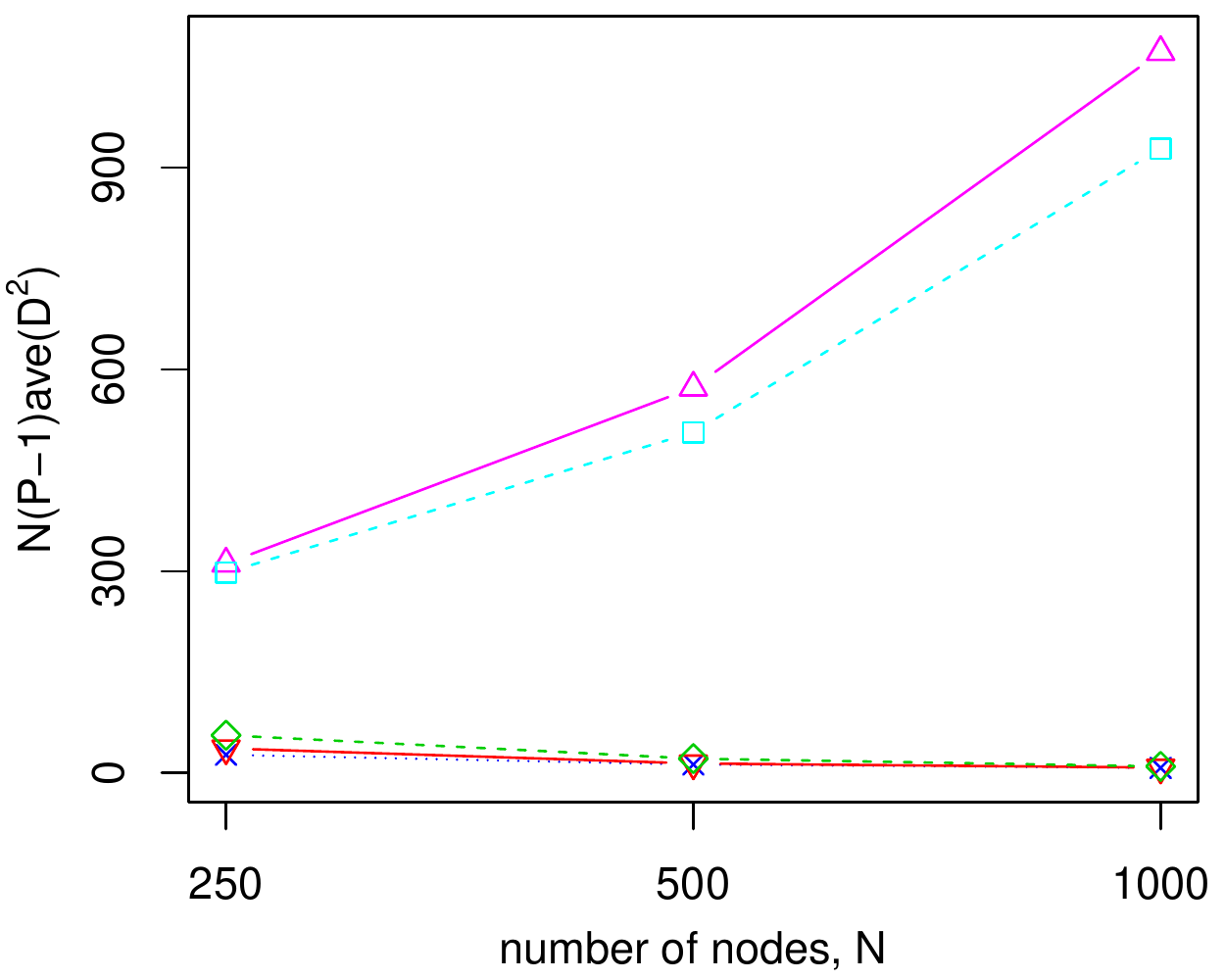}
		\hspace{1.1cm}
		\includegraphics[width=5.7cm]{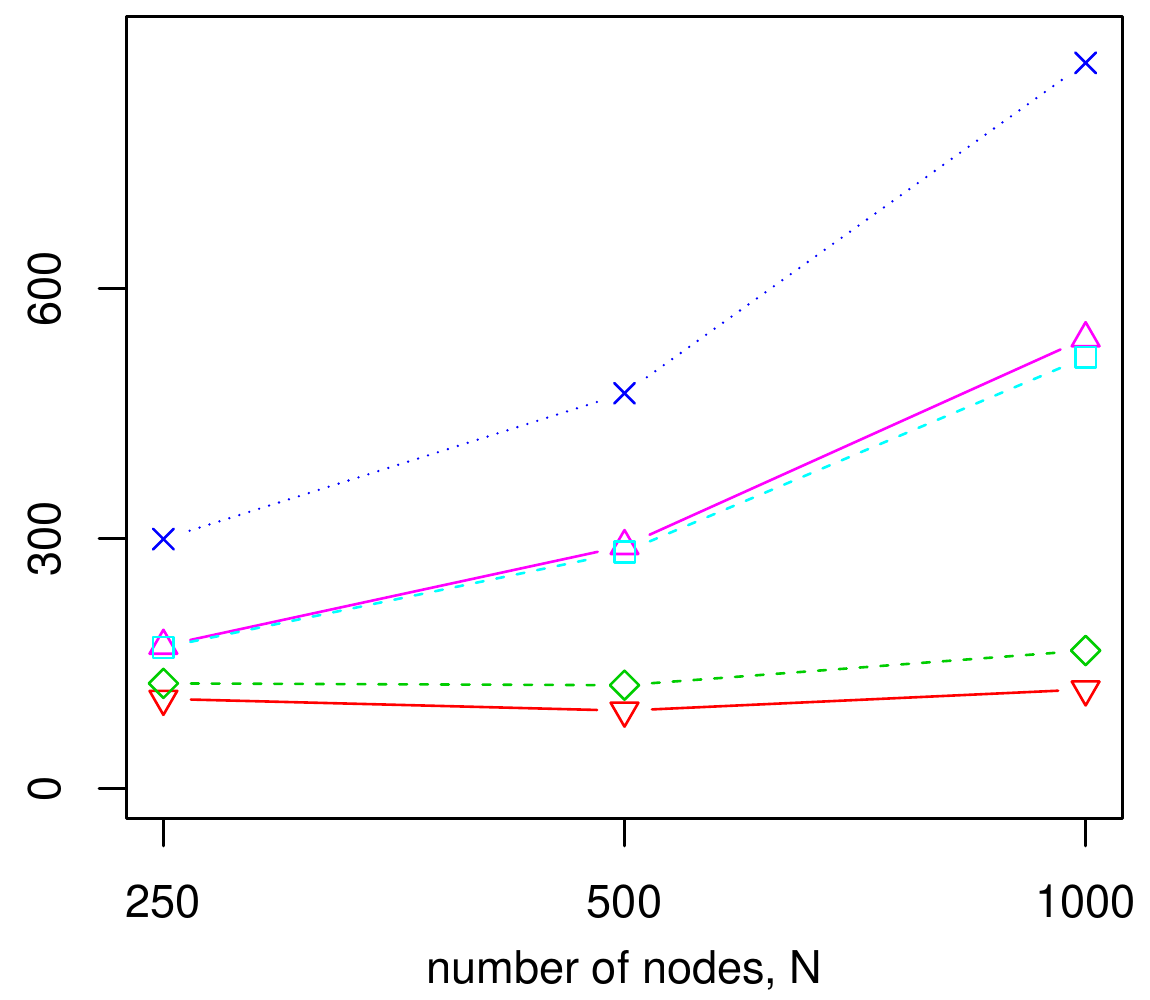} \\
		\includegraphics[width=17cm]{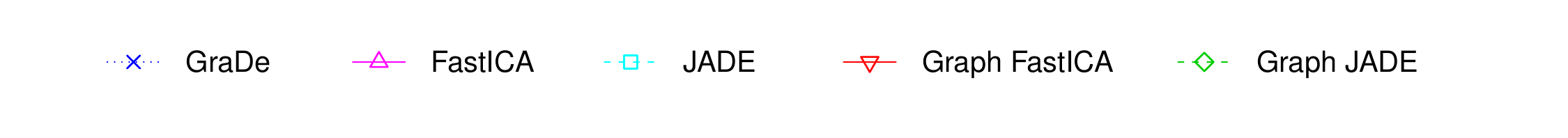}
	\end{minipage}
	\vspace{-1.0cm}
	\caption{The average values of $N(P-1)D^2(\hat{\mathbf{\Gamma}})$ over 1000 repetitions in (M1) top left, (M2) top right, (M3) bottom left), and model 4 (bottom right).}
	\label{fig3}
\end{figure*}

In the three models where the ICs have different adjacency matrices, the ML method is always on average more accurate than GraDe and obviously the oracle version of the ML method outperforms the realistic one. When ICs have the same adjacency matrix in (iv) and $\theta_1 \approx \theta_2$, GraDe is slightly better than the ML estimator. However, then neither of them performs better than random guess estimator. Notice that the CRB bound goes to infinity when $\theta_1 = \theta_2$. On the other hand, in model (iv) when $\theta_2 > \theta_1$, the difference between GraDe and the ML estimator is larger than in models (i)--(iii). The shapes of the CRB curve and the corresponding estimator performance curves are similar. Perhaps the most notable exception is $\text{tr} (\textbf{CRB}_{\Imat_2})$ in model (i) where the estimator curves trend upwards as $\theta_2$ increases whereas the CRB is almost constant for $\theta_2 > 0.18$. Meanwhile, the curves related to $(\textbf{CRB}_{\Imat_2})_{2,2} + (\textbf{CRB}_{\Imat_2})_{3,3}$ have almost identical shapes, so the difference is in the diagonal elements of $\mathbf{\Omega}$, especially in $(\mathbf{\Omega})_{2,2}$ which corresponds to the second IC. Small peaks in the simulation curves cannot be fully explained by the finite number of simulation runs. Interestingly, zooming in shows peaks in CRB curves as well, however, in slightly different places.

It should be reminded that even though the ML estimator is more accurate in these models, GraDe is to be preferred because it does not assume any particular graph signal model and it is not computationally as limited as the ML estimator. However, the gap between the CRB and GraDe performance shows that there is room for improvement. Thus, searching for new graph BSS estimators might be worthwhile.  

\subsection{Graph FastICA and Graph JADE}
To compare Graph FastICA (with $\lambda=0.001$) and Graph JADE (with $\lambda=0.8$) to FastICA, JADE and GraDe, we run simulations for four models and for three graph sizes $N = 250, 500, 1000$. Each model has $P = 4$ ICs which all are GMA(1) signals. 

In model (M1), all ICs have the same adjacency matrix, but $\theta$ and the distribution of $y_i$ in GMA model differ. The components have coefficients $\theta = 0.02, 0.04, 0.06, 0.08$, and follow $t$-distribution with 5, 10 and 15 degrees of freedom, and Gaussian distribution, respectively. This model satisfies the assumptions which JADE, FastICA and GraDe require to consistently estimate the unmixing matrix. Therefore, reasonable performance is expected by all methods.  

In (M2), the ICs share one adjacency matrix, distributions are $t$-distribution with 5 degrees of freedom, uniform distribution, exponential distribution, and Gaussian distribution, and the GMA coefficients are $\theta = 0.05, 0.06, 0.07, 0.08$. The proximity of the GMA coefficients makes the separation difficult for GraDe whereas the distributions are easier to separate for JADE and FastICA. 

The ICs of (M3) have different adjacency matrices and $\theta = 0.05$ for all of them. The distribution of innovations for all ICs is $t_{15}$ which is close to Gaussian. Hence, this setup is challenging for FastICA and JADE, but easy for GraDe. 

(M4) is an example of a case where the unmixing matrix is not identifiable for GraDe, JADE or FastICA, but Graph JADE and Graph FastICA can separate all components. The ICs have identical adjacency matrix, the GMA coefficients are $\theta = 0.04, 0.04, 0.08, 0.08$ and the innovation distributions are $t_{15}$, Gaussian, uniform and Gaussian, respectively. Hence, FastICA and JADE cannot separate second IC from fourth IC, and GraDe cannot separate first IC from second IC or third IC from fourth IC. The results of the simulations are presented in Fig.~\ref{fig3}. 

The average value of $N(P-1)D(\hat{\mathbf{\Gamma}})^2$ in 1000 runs is plotted. In each run, a new adjacency matrix or a set of adjacency matrices is generated so that results on different values of $N$ would be comparable. Otherwise, the realization of the single adjacency matrix could have too large impact on the averages. 

In (M1), where FastICA, JADE and GraDe are almost equally good, Graph JADE and Graph FastICA still demonstrate the best overall performances. As expected, Graph JADE and Graph FastICA also outperform the other methods in (M4). In (M2), where FastICA and JADE are much better than GraDe, Graph JADE achieves the performance of FastICA and JADE, and Graph FastICA even performs better than those. Similarly in (M3), where GraDe outperforms the ICA methods, Graph FastICA and Graph JADE are equally good as GraDe when $N=1000$. When $N$ is smaller, there is a gap between those, but the proposed methods still clearly demonstrate the best overall performance across the models even with small $N$. Comparing Graph FastICA and Graph JADE, the former one is better in all models and for all values of $N$. In the first two models the difference can be partly explained by the difference in performance between the squared symmetric FastICA and JADE, which is larger than that of regular FastICA and JADE. In the third model JADE is slightly better than squared symmetric FastICA. There are two possible explanations why Graph FastICA is still better than Graph JADE. One is that $\lambda=0.001$ of Graph FastICA puts more weight on the graph autocorrelations than $\lambda=0.8$ of Graph JADE, and the other is that Graph FastICA has more efficient downweighting for the non-Gaussianity part when the ICs are nearly Gaussian. 

\begin{table*}[t]
	\centering
	\begin{tabular}{ccccccccccccc}
		\hline
		\vspace{-0.18cm}
		& & & & & & & & & & & &  \\
		& \multicolumn{3}{c}{(M1)} & \multicolumn{3}{c}{(M2)} & \multicolumn{3}{c}{(M3)} & \multicolumn{3}{c}{(M4)} \\ 
		& GF & GJ & GF/GJ & GF & GJ & GF/GJ & GF & GJ & GF/GJ & GF & GJ & GF/GJ
		 \\ 		
		\hline		
		\vspace{-0.18cm}		
        & & & & & & & & & & & &  \\
		$N = 250$ & 0.189 & 0.005 & 41.8 & 0.089 & 0.005 & 18.9 & 0.102 & 0.004 & 24.9 & 0.170 & 0.005 & 37.1 \\
		$N = 500$ & 0.298 & 0.014 & 21.4 & 0.062 & 0.013 & 4.63 & 0.127 & 0.014 & 8.98 & 0.202 & 0.014 & 14.7 \\
		$N = 1000$ & 0.577 & 0.053 & 10.8 & 0.106 & 0.056 & 1.87 & 0.246 & 0.055 & 4.44 & 0.449 & 0.054 & 8.37 \\ \hline
	\end{tabular}
	\caption{Average computation times in seconds for Graph FastICA (GF) and Graph JADE (GJ), and the ratio of the average computation times (GF/GJ).}
	\label{CompTime}
\end{table*}

The computation times of Graph JADE and Graph FastICA methods were recorded in (M1)--(M4) using RStudio Version 1.2.5033 with Quad 3.40 GHz Intel Xeon(R) E3--1230 v5 processor. Graph JADE was faster than Graph FastICA in each model, but the difference depends quite significantly on the model and the size of the graph, as can be seen in Table~\ref{CompTime}. The most plausible explanation for the change in ratios between different graph sizes is that Graph FastICA needs less iterations for large $N$, whereas the number of iterations remains more constant for Graph JADE. When the number of source components increases, it is expected that the Graph FastICA would become faster than Graph JADE, as is the case in comparison between regular FastICA and JADE.

\section{Discussion}
\label{sec:concl}
We derived the CRB for mixing and unmixing matrix estimates in the BSS model, where the ICs are Gaussian graph signals. As specific examples we considered the case of Gaussian GMA signals. We then found out that the only graph BSS estimator currently available, namely the GraDe method, is performing quite far from the bound. The bound however suggests that the mixing matrix could be estimated very accurately. Thus, we also developed and studied numerically an ML estimator in the case that the ICs are Gaussian GMA(1) signals. The performance of ML estimator was found to be in between the CRB and GraDe's performance. However, the ML estimator makes assumptions on the graph signal model which GraDe does not, and it is computationally limited/inefficient with respect to the size of the graph. Hence, interesting open questions are whether the CRB can be reached by any estimator, and thus, is there a practical method which is more efficient than GraDe in the Gaussian graph signal BSS problem.

We have also introduced two BSS methods for non-Gaussian graph signals, which use both non-Gaussianity and graph dependence of the ICs. The proposed Graph JADE uses joint approximate diagonalization of graph autocorrelation and fourth-order cumulant matrices and the proposed Graph FastICA uses fixed-point algorithm to maximize an objective function which is composed of the sum of squares of the diagonal elements of the graph autocorrelation matrices and a user-specified non-Gaussianity measure. For sufficiently large graphs, these methods outperform JADE, FastICA, and GraDe when the sources are non-Gaussian graph signals, and meet the performance of the best method with minimal losses when the sources exhibit only non-Gaussianity or only graph dependence.
%\vfill\pagebreak

%\vspace{-0.1cm}

\appendices

\section{Proof of Theorem \ref{Th_FIM}}\label{App_Th_FIM}
In this appendix, we derive the FIM under the model \eqref{vec_BSS} and parameter vector \eqref{parameter_vector}. According to the Slepian-Bangs formula \cite{SLEPIAN,BANGS} the $(i,j)$th element of the FIM is given by
\be\label{slepian_bangs_cov}
J_{i,j}=\frac{1}{2}{\text{tr}}\left(\Cmat_\xvec^{-1}\frac{\partial\Cmat_\xvec}{\partial\phi_i}\Cmat_\xvec^{-1}\frac{\partial\Cmat_\xvec}{\partial\phi_j}\right), \quad i,j=1,\ldots,P^2 + M.
\ee

Using Kronecker product rules \cite[pp. 60]{MATRIX_COOKBOOK}, we obtain the inverse of $\Cmat_\xvec$, which is given by
\begin{equation*}
\Cmat_\xvec^{-1}=(\Imat_N\otimes\Omegamat^{-\top})\Cmat_\zvec^{-1}(\Imat_N\otimes\Omegamat^{-1})
\end{equation*}
where
\be\label{C_z_inv}
\Cmat_\zvec^{-1}=\sum_{p=1}^{P}\Cmat_p^{-1}\otimes(\evec_p\evec_p^\top).
\ee

It can be observed that $\phi_{P(l-1) + k}=\omega_{k,l}, \quad k,l=1,\ldots,P$. Then we have
\be\label{C_x_deriv_Omega}
\begin{split}
&\frac{\partial \Cmat_\xvec}{\partial\phi_{P(l-1) + k}} = \frac{\partial \Cmat_\xvec}{\partial \omega_{k,l}} = (\Imat_N \otimes \Omegamat) \Cmat_\zvec (\Imat_N \otimes \evec_l \evec_k^\top) + (\Imat_N \otimes \evec_k \evec_l^\top ) \Cmat_\zvec (\Imat_N \otimes \Omegamat^\top), \\
& \qquad k,l =1,\ldots,P
\end{split}
\ee

In addition, for $\phi_{P^2 + \sum_{q=1}^{p-1}M_q + m}=\theta_{p,m}$, we have
\be\label{C_x_deriv_theta}
\begin{split}
\frac{\partial \Cmat_\xvec}{\partial\phi_{P^2 + \sum_{q=1}^{p-1}M_q + m}} &= \frac{\partial \Cmat_\xvec}{\partial \theta_{p,m}} = (\Imat_N \otimes \Omegamat) (\Dmat_{p,m} \otimes \evec_p \evec_p^\top) (\Imat_N \otimes \Omegamat^\top), \\
& \!\!\!\!\!\!\!\!\!\!\!\!\!\!\!\!\!\!\!\!\!\!\!\!\!\!\!\!\!\! m=1,\ldots,M_p,~p =1,\ldots,P,~
\end{split}
\ee
where $\Dmat_{p,m}$ is given in \eqref{C_m_deriv}. The second equality in \eqref{C_x_deriv_Omega} as well as in \eqref{C_x_deriv_theta} is obtained using matrix derivative rules \cite{MATRIX_COOKBOOK}. 

Substituting \eqref{C_x_deriv_Omega} with $\omega_{i,j}$ and $\omega_{k,l}$ into \eqref{slepian_bangs_cov}, we obtain
\be\label{FIM_compute_Omega}
\begin{split}
&J_{P(j-1) + i,P(l-1) + k} = N{\text{tr}} \left( \evec_i \evec_j^\top \Omegamat^{-1} \evec_k \evec_l^\top \Omegamat^{-1} \right) \\
&+{\text{tr}} \left( \Cmat_\zvec (\Imat_N \otimes (\evec_j \evec_i^\top \Omegamat^{-\top})) \Cmat_\zvec^{-1} (\Imat_N \otimes (\Omegamat^{-1} \evec_k \evec_l^\top))\right) .
\end{split}
\ee

Using \eqref{FIM_compute_Omega}, after some algebraic computations, we can find the $(j,l)$th block in $\Jmat_\Omegamat$ which corresponds to the $j$th and $l$th columns of $\Omegamat$. This block is given by
\be\label{FIM_compute_Omega_next}
\begin{split}
\Jmat_\Omegamat^{(j,l)}&=\Omegamat^{-\top} {\rm{E}} \{(\zvec_j^\top \otimes \Imat_P) \Cmat_\zvec^{-1} (\zvec_l \otimes \Imat_P)\} \Omegamat^{-1} + \Omegamat^{-\top} N \evec_l \evec_j^\top \Omegamat^{-1}, \quad j,l=1,\ldots,P.
\end{split}
\ee

From the statistical independence of the zero-mean ICs, $\zvec_p,~p=1,\ldots,P$, and using the expression for $\Cmat_\zvec^{-1}$ from \eqref{C_z_inv}, we can write
\be\label{FIM_compute_z_part}
\begin{split}
&{\rm{E}}\{(\zvec_j^\top \otimes \Imat_P) \Cmat_\zvec^{-1} (\zvec_l \otimes \Imat_P) \} = \delta(j-l) \sum_{p=1}^{P} \kappa_{j,p} \evec_p \evec_p^\top =\delta(j-l) {\text{diag}} (\kappa_{j,1}, \ldots, \kappa_{j,P}), \\
& \qquad j,l =1 ,\ldots P
\end{split}
\ee
where $\kappa_{j,p}$ is defined in \eqref{kappa_define}. 

Substituting \eqref{FIM_compute_z_part} into \eqref{FIM_compute_Omega_next} yields
\be\label{FIM_compute_Omega_sub}
\begin{split}
&\Jmat_\Omegamat^{(j,l)} = \Omegamat^{-\top} (\delta(j-l) {\text{diag}} (\kappa_{j,1}, \ldots, \kappa_{j,P}) + N \evec_l \evec_j^\top) \Omegamat^{-1}, \quad j,l=1,\ldots,P .
\end{split}
\ee
Rewriting \eqref{FIM_compute_Omega_sub} and replacing $(j,l)$ with $(i,j)$, we obtain \eqref{FIM_compute_Omega_block}.\\
\indent

Substituting \eqref{C_x_deriv_theta} with $\theta_{p,i}$ and $\theta_{q,j}$ into \eqref{slepian_bangs_cov} and using trace and Kronecker product properties, we obtain \eqref{FIM_compute_theta_sub}. Similarly, substituting \eqref{C_x_deriv_Omega} and \eqref{C_x_deriv_theta} with $\omega_{i,j}$ and $\theta_{p,m}$, respectively, into \eqref{slepian_bangs_cov} and using trace and Kronecker product properties, we also obtain
\be\label{FIM_compute_Omega_theta_element}
\begin{split}
&j_{P(j-1) + i,P^2 + \sum_{q=1}^{p-1}M_q + m} = \delta(p-j) [\Omegamat^{-1}]_{p,i} {\text{tr}} (\Cmat_p^{-1} \Dmat_{p,m}), \quad  m=1,\ldots,M_p,~i,j,p=1,\ldots,P ,
\end{split}
\ee
where $[\Omegamat^{-1}]_{p,i}$ is the $(p,i)$th element of $\Omegamat^{-1}$ or equivalently, the $(i,p)$th element of $\Omegamat^{-\top}$. Using \eqref{FIM_compute_Omega_theta_element} after some algebraic computations, we get that the $P\times M_p$ block in the FIM corresponding to the $j$th column of $\Omegamat$ and to $\thetavec_p$ is given by $\delta(p-j)\Omegamat^{-\top}\evec_p\svec_p^\top,~j,p=1,\ldots,P$. Consequently, we obtain that $\Jmat_{\Omegamat,\thetavecsmall}$ is a block diagonal matrix, whose $p$th diagonal block is a $P\times M_p$ matrix given by \eqref{FIM_compute_Omega_theta_block}, $p = 1, \ldots, P$. Thus, the $(P^2 + M)\times(P^2 + M)$ FIM is derived. 

\section{Proof of Proposition \ref{prop_CRB_omega}}\label{App_prop_CRB_omega}
In this appendix, we derive the CRB for estimation of $\Omegamat$. The block matrix inversion formula \cite{MATRIX_COOKBOOK} gives
\begin{equation*}
{\textbf{CRB}}_\Omegamat=(\Jmat_\Omegamat-\Jmat_{\Omegamat,\thetavecsmall}\Jmat_\thetavecsmall^{-1}\Jmat_{\Omegamat,\thetavecsmall}^\top)^{-1}.
\end{equation*}
Substituting \eqref{FIM_compute_theta_sub} and \eqref{FIM_compute_Omega_theta_block} into the term $\Jmat_{\Omegamat, \thetavecsmall} \Jmat_\thetavecsmall^{-1} \Jmat_{\Omegamat, \thetavecsmall}^\top$, we obtain a block diagonal matrix, whose $p$th diagonal block is
\be\label{first_term}
\begin{split}
&(\Jmat_{\Omegamat, \thetavecsmall} \Jmat_\thetavecsmall^{-1} \Jmat_{\Omegamat, \thetavecsmall}^\top)^{(p)} = \Omegamat^{-\top}(\svec_p^\top\Jmat_{\thetavecsmall_p}^{-1}\svec_p)\evec_p \evec_p^\top \Omegamat^{-1}, \quad p=1,\ldots,P.
\end{split}
\ee
Then, substituting \eqref{FIM_compute_Omega_block} and \eqref{first_term} into the term $\Jmat_\Omegamat - \Jmat_{\Omegamat, \thetavecsmall} \Jmat_\thetavecsmall^{-1} \Jmat_{\Omegamat, \thetavecsmall}^\top$, we obtain $\textbf{CRB}_{\Omegamat}^{-1}$, whose $(i,j)$th block is given by
\begin{equation*}
\begin{split}
&(\textbf{CRB}_{\Omegamat}^{-1})^{(i,j)} = \begin{cases}
\Omegamat^{-\top}(\zeta_i \evec_i \evec_i^\top + \sum_{\underset{l \neq i}{l=1}}^{P} \kappa_{i,l} \evec_l \evec_l^\top) \Omegamat^{-1}, \quad i=j\\
\Omegamat^{-T}N\evec_j\evec_i^\top\Omegamat^{-1}, \quad i\neq j \\
\end{cases}, \quad i,j = 1, \ldots, P
\end{split}
\end{equation*}
where $\zeta_i$ is defined in \eqref{zeta_general}.\\
\indent
Using a matrix version of Cauchy--Schwarz inequality \cite[Eq. (2.3)]{PECARIC} with arbitrary matrices, $\Amat_1$ and $\Amat_2$, we can write
\be\label{CS_general}
\Amat_1\Amat_1^\top\succeq\Amat_1\Amat_2^\top(\Amat_2\Amat_2^\top)^{-1}\Amat_2\Amat_1^\top
\ee
Substituting $\Amat_1 = {\text{vec}}^\top(\Imat_N)$ and $\Amat_2^\top$ to be a matrix, whose $m$th column is $(\Cmat_p^{-\frac{1}{2}} \otimes \Cmat_p^{-\frac{1}{2}}) {\text{vec}}(\Dmat_{p,m}),$ $m=1,\ldots,M_p$, in \eqref{CS_general} and applying trace, vectorization, and Kronecker product rules, we can show that
\be
2N \geq \svec_p^\top\Jmat_{\thetavecsmall_p}^{-1}\svec_p\quad p=1,\ldots,P. \notag
\ee
Hence, $\zeta_p\geq 0,~\forall p=1,\ldots,P$. To invert $\textbf{CRB}_{\Omegamat}^{-1}$, some algebraic manipulations and Kronecker product rules are applied, in a similar manner to the compact CRB derivation in \cite{ESA_CRB} and to the CRB derivation in \cite{YEREDOR_GAUSSIAN_CRB}. Finally, the $(i,j)$th block of $\textbf{CRB}_{\Omegamat}$ is obtained as given in \eqref{CRB_compute_Omega_block}.

\section{Proof of Corollary \ref{coro_CRB_gamma}}\label{App_coro_CRB_gamma}
In this appendix, we derive the CRB for estimation of the unmixing matrix $\Gammamat$. To derive $\textbf{CRB}_{\Gammamat}$, we use the CRB formula for estimating a function $\gvec(\phivec)={\text{vec}}(\Omegamat^{-1})$. Generally, the CRB for unbiased estimation of $\gvec(\phivec)$ is given by \cite{KAY,TTB}
\be\label{CRB_transformation}
{\textbf{CRB}}_{\gvec(\phivec)} = \frac{\ud \gvec( \phivec)}{\ud \phivec}{\textbf{CRB}}_{\phivec} \frac{\ud \gvec(\phivec)}{\ud \phivec}^\top
\ee
where the $(i,j)$th element of the matrix, $\frac{\ud\gvec(\phivecsmall)} {\ud \phivecsmall}$, is $\left [\frac{\ud\gvec(\phivecsmall)} {\ud \phivecsmall} \right]_{i,j} = \frac{\partial g_i (\phivecsmall)}{\partial \phi_{j}}$. 

Using matrix derivative rules \cite{MATRIX_COOKBOOK}, we obtain
\be\label{vec_inv_Omega_element_deriv}
\frac{\partial}{\partial\omega_{i,j}} {\text{vec}} (\Omegamat^{-1}) = -{\text{vec}} (\Omegamat^{-1} \evec_i \evec_j^\top \Omegamat^{-1}).
\ee
Moreover, using \eqref{vec_inv_Omega_element_deriv} and Kronecker product rules, we have
\be\label{vec_inv_Omega_mat_deriv}
\frac{\ud {\text{vec}} (\Omegamat^{-1})}{\ud {\text{vec}} (\Omegamat)} = -\Omegamat^{-\top} \otimes \Omegamat^{-1} = -(\Omegamat^{-\top} \otimes \Imat_P) (\Imat_P \otimes \Omegamat^{-1}).
\ee

Substituting $\gvec(\phivec)={\text{vec}}(\Omegamat^{-1})$ and \eqref{vec_inv_Omega_mat_deriv} into the term $\frac{\ud \gvec(\phivec)} {\ud\phivec}$, we obtain
\be\label{g_deriv}
\frac{\ud \gvec(\phivec)}{\ud \phivec} = -[(\Omegamat^{-\top} \otimes \Imat_P)(\Imat_P \otimes \Omegamat^{-1}),\zerovec].
\ee
Finally, substitution of \eqref{equivariance_Omega} and \eqref{g_deriv} into \eqref{CRB_transformation} yields \eqref{CRB_Gammamat}.

% References should be produced using the bibtex program from suitable
% BiBTeX files (here: strings, refs, manuals). The IEEEbib.bst bibliography
% style file from IEEE produces unsorted bibliography list.
% -------------------------------------------------------------------------
%\newpage

\label{sec:ref}

\end{document}